\documentclass[runningheads]{llncs}
\usepackage{graphicx}
\usepackage{amsmath}
\usepackage{amsfonts}
\usepackage{appendix}
 \usepackage{booktabs} 
 \usepackage{cite}
\usepackage{tabularx}
\usepackage{makecell}
\usepackage{xcolor}
\usepackage{subfig}
\usepackage{comment} 
\usepackage{multirow}
\usepackage{threeparttable}
\usepackage{graphicx}
\newsavebox{\measurebox}
\let\llncssubparagraph\subparagraph
\let\subparagraph\paragraph
\usepackage[compact]{titlesec}
\usepackage{hyperref}
\let\subparagraph\llncssubparagraph

\newcommand{\myparagraph}[1]{\vspace{5pt}\noindent\textbf{#1}}

\newcommand{\paced}{$(b, N, \mathcal{T}, \mathcal{D})\ $}
\newcommand{\pt}{$\mathcal{T} = (c, t, r, \alpha, q, d)\ $}
\newcommand{\pd}{$\mathcal{D} = (\gamma, \eta, \lambda)\ $}

\newcommand{\h}{\textsf{H}}
\newcommand{\outcomm}{\lceil cn\rceil}
\newcommand{\allo}{\textit{All-O}}
\newcommand{\allc}{\textit{All-C}}
\newcommand{\alld}{\textit{All-D}}

\newcommand{\POM}{$\mathsf{POM}\ $}

\begin{document}
%
%\title{Contribution Title\thanks{Supported by organization x.}}
\title{ACeD: Scalable Data Availability Oracle}
%\begin{comment}

\author{Peiyao Sheng\inst{1}\thanks{The first two authors contributed equally to this work.}
\and
Bowen Xue\inst{2}{$^\star$}
 \and
Sreeram Kannan \inst{2}%\orcidID{2222--3333-4444-555
\and 
 Pramod Viswanath\inst{1}%\orcidID{0000-1111-2222-3333} 
}
\authorrunning{P.\ Sheng et al.}
\institute{University of Illinois, Urbana-Champaign IL 61801  \and
University of Washington at Seattle, WA}
%\end{comment}
\maketitle              % typeset the header of the contribution

\begin{abstract}
    A popular method in practice offloads computation and storage in blockchains by relying on committing only hashes of off-chain data into the blockchain. This mechanism is acknowledged to be vulnerable to a stalling attack:  the blocks corresponding to the committed hashes may be unavailable at any honest node. The straightforward solution of broadcasting all blocks to the entire network sidesteps this data availability attack, but it is not scalable.  In this paper, we propose  ACeD, a  scalable solution to this data availability problem with $O(1)$ communication efficiency, the first to the best of our knowledge. 
    The key innovation is a new protocol that requires each of the $N$ nodes to receive only $O(1/N)$ of the block, such that the data is guaranteed to be available in a distributed manner in the network. Our solution creatively integrates coding-theoretic designs inside of Merkle tree commitments to guarantee efficient and tamper-proof reconstruction; this solution is distinct from Asynchronous Verifiable Information Dispersal \cite{cachin2005asynchronous} (in guaranteeing efficient proofs of malformed coding) and Coded Merkle Tree \cite{yu2020coded} (which only provides guarantees for random corruption as opposed to our guarantees for worst-case corruption). We implement ACeD with full functionality in 6000 lines of {\sf Rust} code, integrate the functionality as a smart contract into Ethereum via a high-performance implementation demonstrating up to 10,000 transactions per second in throughput and 6000x reduction in gas cost on the Ethereum testnet Kovan. Our code is available in \cite{aced}.%\cite{our-code}. 

%     Blockchains 
    
%     Scalabilty of secure blockhains has been a major research area in the past years with significant progress. One mechanism is for a smaller blockhaito derive security 

%     However the deployment of these solutions either requires the genesis of new blockchains or hard forks in existing blockchains.
%   Off-chain protocols derive trust from  core blockchain consensus while  scaling overall  performance (throughput). Payment channels (more generally state channels) are an instance of off-chain scaling protocols. \pv{We need to say what state channels do and how they achieve their performance -- esp limitations. This will lead to highlight/contrast  the approach in this paper. Basically, we scale up the light-node functionality }  
\end{abstract}

%
%\begin{abstract}
%While blockchains aim to revolutionize service and market in a decentralized manner, several hurdles can not be circumnavigated, chief among them is the scalability problem. The new generation of sharding solutions and base-layer consensus protocols addressing the scaling problem requires redesigning the core chain. Can we scale existing chains using an Oracle model? We propose ACeD, a data availability oracle on top of any smart contract platform. The key innovation is to construct a data structure that allows the dissemination of portions of information through the network such that the data can be reconstructed from the honest nodes on demand. We show how to integrate this oracle as a smart contract into Ethereum via a high-performance implementation, along with an incentive design that ensures that users do not deviate from the protocol. 

%\keywords{Blockchain Scaling  \and Second keyword \and Another keyword.}
%\end{abstract}

\section{Introduction}
    Public blockchains such as Bitcoin and Ethereum have demonstrated themselves to be secure in practice (more than a decade of safe and live operation in the case of Bitcoin), but at the expense of poor performance (throughput of a few transactions per second and hours of latency). Design of high performance (high throughput and low latency) blockchains without sacrificing security has been a major research area in recent years, resulting in new proof of work \cite{bagaria2019prism,yang2019prism,yu2020ohie,fitzi2018parallel}, proof of stake \cite{gilad2017algorand,daian2019snow,kiayias2017ouroboros,david2018ouroboros,badertscher2018ouroboros}, and hybrid \cite{buterin2017casper,pass2017hybrid} consensus protocols. These solutions entail a wholesale change to the core blockchain stack and existing blockchains can only potentially upgrade with very significant practical hurdles (e.g.: hard fork of existing ledger). To address this concern, high throughput scaling solutions are explored via ``layer 2" methods, including payment channels \cite{decker2015fast,miller2017sprites} and state channels \cite{poon2017plasma,teutsch2019scalable, kalodner2018arbitrum}. These solutions involve ``locking" a part of the ledger on the blockchain and operating on this trusted, locked state on an application layer outside the blockchain; however the computations are required to be semantically closely tied to the blockchain (e.g.: using the same native currency for transactions) and  the locked nature of the ledger state leads to limited applications (especially, in a smart contract platform such as Ethereum). % Para-1: Scalability bottleneck, existing solutions. 
    %-  Scaling throughput performance with the number of nodes is a central problem in blockchains.
    %-  Several solutions explored including on-chain solutions such as improvements in consensus and Sharding. Improvements in conensus while providing better utilization of bandwidth all eventually face the replication barrier: i.e., all nodes in the network need to download and validate all the data. Sharding solutions are still unable to deal with the full range of post facto bribing attacks. Layer-2 solutions such as Payment and state channels lock up tokens / state from the blockchain so that a group of nodes can process that state off-line. However, these lead to a limited set of useage scenarios since the state locked by the state channel is not accessible by other nodes.
    %\vspace{-4mm}
   \begin{figure}
\centering
\sbox{\measurebox}{%
    \begin{minipage}[b]{.43\textwidth}
    \centering
    \subfloat
    []
    {\label{fig:two-layer}\includegraphics[width=\textwidth]{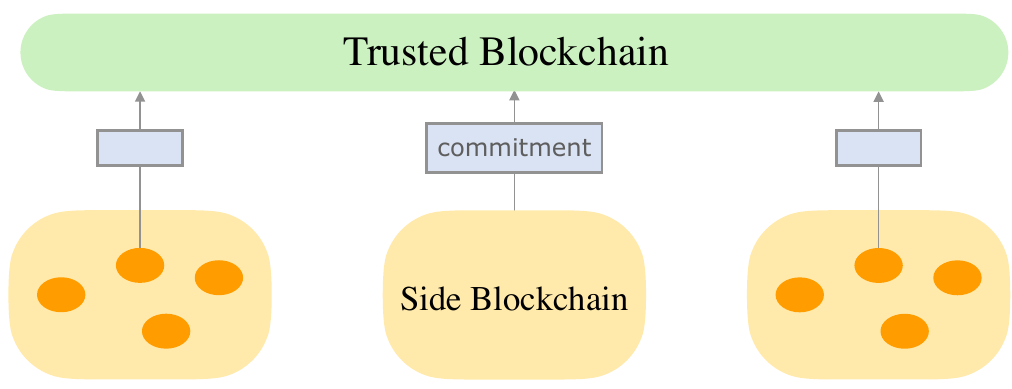}}

    \subfloat
    []
    {\label{fig:three-layer}\includegraphics[width=\textwidth]{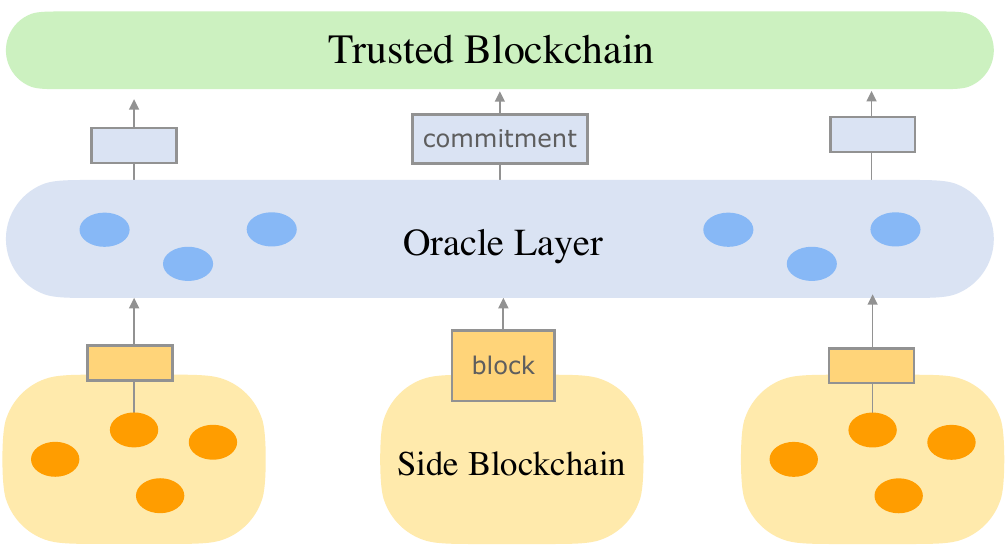}}
    \end{minipage} 
}
\usebox{\measurebox}
\begin{minipage}[b][\ht\measurebox]{.51\textwidth}
  \subfloat
    []
    {\label{fig:ACeD}\includegraphics[width=\textwidth]{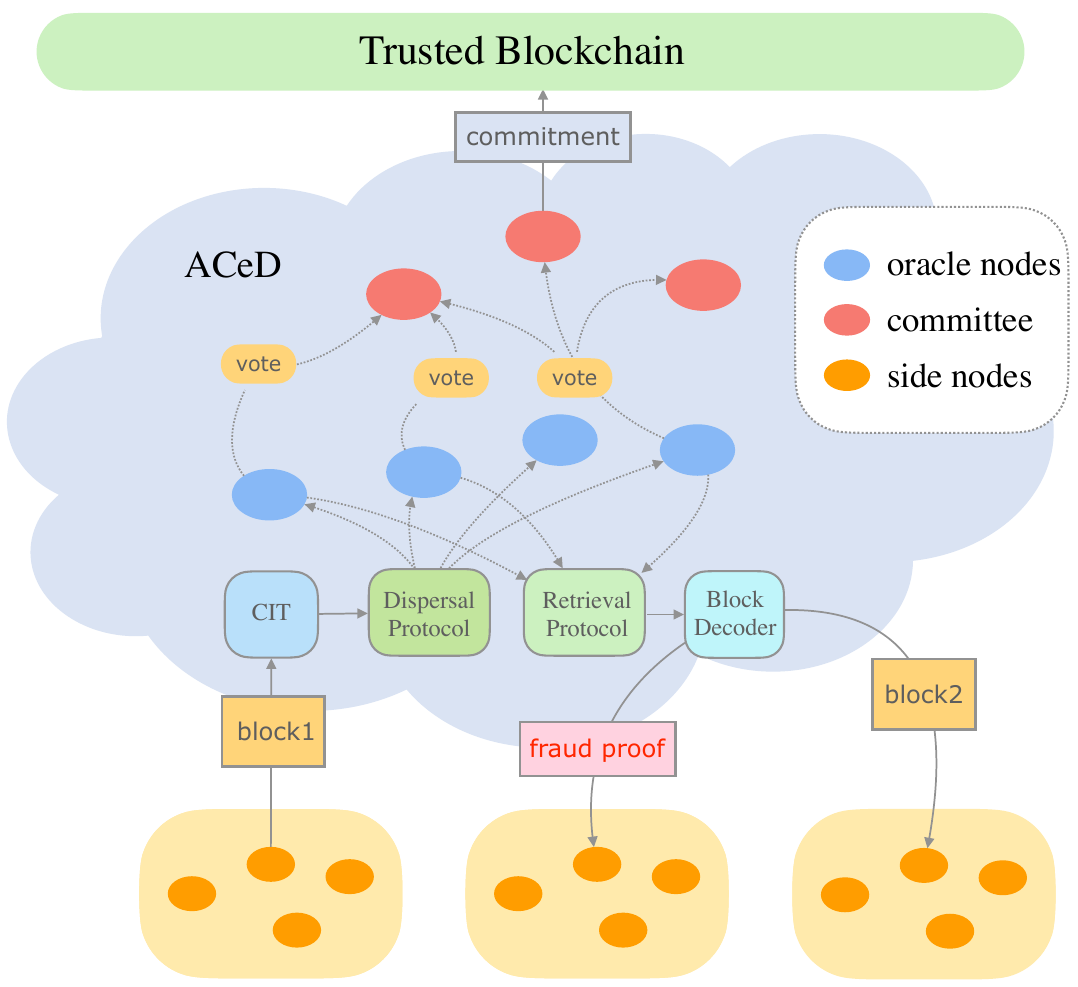}}
  \end{minipage}
\caption{(a) Side blockchains commit the hashes of blocks to a larger trusted blockchain. (b) An oracle layer is introduced to ensure data availability. (c) ACeD is a data availability oracle.}
\end{figure}    
     %\vspace{-4mm}
   
    In practice, a popular form of scaling blockchain performance is via the following: a smaller blockchain  (henceforth termed ``side blockchain") derives trust from a larger blockchain (henceforth termed ``trusted blockchain") by committing the hashes of its blocks periodically to the trusted  blockchain (Fig.\ref{fig:two-layer}). 
    The ordering of blocks in the side blockchain is now determined by the order of the hashes in the trusted blockchain; this way the security of the side blockchain is directly derived from that of the trusted blockchain. This mechanism is simple, practical and efficient -- a single trusted blockchain can cheaply support a large number of side blockchains, because it does not need to store, process or validate the semantics of the blocks of the side blockchains, only storing the hashes of them. It is also very popular in practice, with several side blockchains running on both Bitcoin and Ethereum; examples include donation trace (Binance charity \cite{foti2020blockchain}) and  diplomas and credentials verification (BlockCert used by MIT among others \cite{jirgensons2018blockchain}). 
    
   % Para-3: Stalling attack and naive solution
 For decentralized operations of this simple scaling mechanism,  any node in the side blockchain should be able to commit hashes to the trusted blockchain. This simple operation, however, opens up a serious vulnerability:  an adversarial side blockchain node can  commit the hash of a block without transmitting the block to any other side blockchain node. Thus, while the hash is part of the ordering according to the trusted blockchain, the block corresponding to the hash is itself unavailable at any side blockchain node; this {\em data availability attack} is a serious threat to the liveness of the side blockchain.  
     The straightforward solution to this  data availability attack is to store all blocks on the trusted blockchain,  but the communication and storage overhead for the trusted blockchain is directly proportional to the size of the side blockchains and the mechanism is no longer scalable.  
    
  We propose an intermediate ``data availability oracle" layer that interfaces between the side blockchains and the trusted blockchain (Fig. \ref{fig:three-layer}). The oracle layer  accepts blocks from side blockchains,  pushes  verifiable commitments to the trusted blockchain and  ensures  data availability to the side blockchains. The  $N$ oracle layer nodes  work together to reach a consensus about whether the proposed block is retrievable (i.e., data is available) and only then commit it to the trusted blockchain. The key challenge is how to securely and efficiently share the data amongst the oracle nodes to verify data availability; if all oracle nodes maintain a copy of the entire side blockchain data locally (i.e.,  repetition),  then that obviously leads to simple majority vote-based retrievability but is not scalable. If the data is dispersed among the nodes to reduce redundancy (we call this ``dispersal"), even one malicious oracle node can violate the retrievability. Thus it appears there is an inherent trade-off between  security and efficiency. 
    
The main point of this paper is to demonstrate that the trade-off between security and efficiency is not inherent; the starting point of our solution is the utilization of an erasure code such that different oracle nodes receive different coded chunks. A key issue is how to ensure the integrity and correctness of the coded chunks. Intuitively we can use a Merkle tree to provide the proof of inclusion for any chunk, but a malicious block producer can construct a Merkle tree of a bunch of nonsense symbols so that no one can successfully reconstruct the block. To detect such attacks, nodes can broadcast what they received and meanwhile download chunks forwarded by others to decode the data and check the correctness. Such a method is used in an asynchronous verifiable information dispersal (AVID) protocol proposed by \cite{cachin2005asynchronous}. AVID makes progress on storage savings with the help of erasure code but sacrifices  communication efficiency. Nodes still need to download all data chunks;  hence, the communication complexity is $O(Nb)$. An alternative approach to detect incorrect coding attacks is via an {\em incorrect-coding proof} (also called fraud proof), which contains symbols that fail the parity check and can be provided by any node who tries to reconstruct the data;   consider using an $(n,k)$ Reed-Solomon code (which is referred as 1D-RS), with $k$ coded symbols in the fraud proof, essentially not much better than downloading the original block ($n$ symbols). For reducing the size of fraud proof, 2D-RS \cite{albassam2019fraud} places a block into a $(\sqrt{k}, \sqrt{k})$ matrix and apply $(\sqrt{n},\sqrt{k})$ Reed-Solomon code on all columns and rows to generate $n^2$ coded symbols; 2D-RS reduces the fraud proof size to $O(\sqrt{b}\log b)$ if we assume symbol size is constant.

 In summary (Table~\ref{tab:comparison}), to find a scalable solution for the data availability oracle problem, erasure code based methods must defend against the incorrect coding attack while minimizing the storage and communication cost. 1D-RS has low communication complexity but when the storing node is adversarial, the storage and download overhead are factor $b$ worse than optimal. The storage and download overhead of AVID remains optimal even under adversarial storing node but the communication complexity is factor $N$ worse than optimal. A full analysis on each performance entry in the table is provided in Appendix \ref{sec:performance-table}.

Our main technical contribution is  a new protocol, called Authenticated Coded Dispersal (ACeD), that provides a scalable solution to the data availability oracle problem.  ACeD achieves near-optimal performance on all parameters: optimal storage, download and communication overhead under the normal path, and near-optimal (worse by a logarithmic factor) storage and download overhead when the storing node is adversarial, cf.\ Table~\ref{tab:comparison}. %In ACeD, each of the $N$ nodes in the network only needs to store $O(1/N)$ of the block in a distributed manner to guarantee secure retrievability of data from the whole oracle network.  
We state and  prove the security of the data availability guarantee and efficiency properties of ACeD in a formal security model (Section~\S\ref{sec:system-overview}).

\begin{table}[]
\centering
\caption{Performance metrics for different data availability oracles ($N$: number of oracle nodes, $b$: block size).}
\begin{threeparttable}
\resizebox{\textwidth}{!}{
\begin{tabular}{|l|l|c|c|c|c|c|}
\hline
\multirow{3}{*}{}    & maximal   & \multicolumn{2}{c|}{normal case}                              & \multicolumn{2}{c|}{worst case}                               & \multirow{3}{*}{\begin{tabular}[c]{@{}c@{}}communication \\ complexity\end{tabular}} \\ \cline{3-6}
                     & adversary & storage                       & download                      & storage                       & download                      &                                                                                      \\
                     & fraction  & \multicolumn{1}{l|}{overhead} & \multicolumn{1}{l|}{overhead} & \multicolumn{1}{l|}{overhead} & \multicolumn{1}{l|}{overhead} &                                                                                      \\ \hline
uncoded (repetition) & 1/2       & $O(N)$                        & $O(1)$                        & $O(N)$                        & $O(1)$                        & $O(Nb)$                                                                              \\ \hline
uncoded (dispersal)  & 1/N       & $O(1)$                        & $O(1)$                        & $O(1)$                        & $O(1)$                        & $O(b)$                                                                               \\ \hline
AVID \cite{cachin2005asynchronous}                & 1/3       & $O(1)$                        & $O(1)$                        & $O(1)$                        & $O(1)$                        & $O(Nb)$                                                                              \\ \hline
1D-RS                & 1/2       & $O(1)$                        & $O(1)$                        & $O(b)$                        & $O(b)$                        & $O(b)$                                                                               \\ \hline
2D-RS \cite{albassam2019fraud}                & 1/2       &     $O(1)$                    &     $O(1)$                      &          $O(\sqrt{b}\log b)$      &          $O(\sqrt{b}\log b)$                 &      $O(b)$                                                                        \\ \hline
ACeD                 & 1/2       & $O(1)$                        & $O(1)$                        & $O(\log b)$                   & $O(\log b)$                   & $O(b)$                                                                          \\ \hline
\end{tabular}}
%\begin{tablenotes}\footnotesize
%\item[*] AVID is an asynchronous protocol.
%\item[\textdagger] \bx{assuming symbol size (shareSize \cite{albassam2019fraud}) is constant} \pramod{you can get rid of all these notes from the table. Just put the "metrics definitions are in Table 2" into the regular caption of this table.}
%\item See metrics definition in Table~\ref{tab:metrics}.
%\end{tablenotes}
\end{threeparttable}

\label{tab:comparison}
\end{table}
%why download efficient for ACeD be log n? I thought it is only a constant fraction of block like gamma b .
\begin{table}
\centering
\caption{System Performance Metrics}
\label{tab:metrics}
\renewcommand{\tabularxcolumn}{m} 
\resizebox{\textwidth}{!}{
\begin{tabularx}{\textwidth}{|c|c|>{\raggedright}X|}
\hline
\textbf{Metric} & \textbf{Formula}&\textbf{Explanation} \tabularnewline
\hline
Maximal adversary fraction & $\beta$ & The maximum number of adversaries is $\beta N$.
\tabularnewline
\hline
Storage overhead & $D_{\text{store}} / D_{\text{info}}$& The ratio of total storage used and total information stored.
 \tabularnewline
\hline
Download overhead & $D_{\text{download}}/D_{\text{data}}$ & The ratio of the size of downloaded data and the size of reconstructed data.
\tabularnewline
\hline
Communication complexity &$D_{\text{msg}} $ &Total number of bits communicated.
\tabularnewline
\hline
\end{tabularx}}
\end{table}

     \noindent {\bf Technical summary of ACeD}. There are four core components in ACeD, as is shown in Fig.\ref{fig:ACeD} with the following highlights. 
     \begin{itemize}
         \item ACeD develops a novel coded commitment generator called {\em Coded Interleaving Tree}  (CIT), which is constructed layer by layer in an interleaved manner embedded with erasure codes. The interleaving property avoids downloading extra proof and thus minimizes the number of symbols needed to store. 
         \item A  dispersal protocol is designed to disperse tree chunks among the network with the least redundancy and we show how feasible dispersal algorithms ensure the reconstruction of all data. 
         \item A hash-aware peeling decoder is used to achieve linear decoding complexity. The fraud proof is minimized to a {\em single parity equation}.
     \end{itemize}

  \noindent  \textbf{Performance  guarantees of ACeD.}  Our main mathematical claim is that safety of ACeD holds as long as the trusted blockchain is secure, and ACeD is live as long as the trusted blockchain is live and a majority of oracle nodes are honest (i.e., follow protocol) (Section~\S\ref{sec:securityanalysis}). ACeD is the first scalable data availability oracle that promises storage and communication efficiency while providing a guarantee for security with a provable bound and linear retrieval complexity; see Table \ref{tab:comparison} with details deferred to Section~\S\ref{sec:efficiencyanalysis}. The block hash commitment on the trusted blockchain and the size of fraud proof are both in constant size.

    %Given an input block, by constructing CIT, we get a verifiable hash commitment and a list of coded symbols. These symbols are further dispersed among the oracle network together with the proof of membership so that any node from side blockchains can retrieve and reconstruct the whole block afterwards or generate a fraud proof when suffering attack. 
    
 \noindent {\bf Incentives}. From a rational action point of view,  oracle nodes are naturally incentivized to mimic others' decisions without storing/operating on  their own data. This ``information cascade" phenomenon is recognized as a basic challenge of  actors  abandoning their own information in favor of inferences based on actions of earlier people when making sequential decisions \cite{easley2010networks}. In the context of ACeD, we carefully use the semantics of the data dispersal mechanisms to design a probabilistic auditing mechanism that ties the vote of data availability to an appropriate action  by any oracle node. This allows us to create a formal rational actor model where the incentive mechanism can be mathematically modeled: we show that the honest strategy is a strong 
 %In that case, an oracle node may vote that some block is available but can not provide any data when a retrieval request is received. To maintain system functionality when participants are rational, our dispersal protocol firstly stipulates the partitions each oracle node ought to store so that a specific node can only claim a block is available if it receives all the assigned data. Since different nodes are assigned with distinct subset of data, the decision of one node can not be borrowed by any other node reasonably. If someone still wants to do that, we design an incentive mechanism that contains an audit scheme. If a node doesn't receive enough data but follow other's decision and vote for yes, it will be audited with a certain probability. The node then needs to provide the proper data or lose the stake. We prove that with proposed mechanism, a Bayesian 
 Nash equilibrium;  the details are deferred to Appendix~\S\ref{sec:incentiveanalysis}. 
    
    \noindent {\bf Algorithm to system design and implementation}. We design an efficient system architecture implementing the ACeD components. Multithreaded erasure code encoding and block pipelining designs significantly parallelize the operations leading to a high performing architecture. We implement this design in roughly 6000 lines of code in {\sf Rust} and integrate ACeD with Ethereum (as the trusted blockchain).  We discuss the design highlights and implementation optimizations in Section \S\ref{sec:implementation}. 
   
   \noindent {\bf Evaluation}. ACeD is very efficient theoretically, but also in practice.  Our implementation of ACeD is run by lightweight  oracle nodes (e.g.: up to 6 CPU cores) communicating over  a wide area network (geographically situated in three continents) and is integrated  with the Ethereum testnet {\sf Kovan} with full functionality for the side blockchains to run Ethereum smart contracts. Our implementation scales a side blockchain throughput up to 10,000 tx/s while adding  a latency of only a few seconds. Decoupling computation from Ethereum (the trusted blockchain) significantly reduces gas costs associated with side blockchain transactions: in our experiments on a popular Ethereum app Cryptokitties, we find that the gas (Ethereum transaction fee) is reduced by a factor of over 6000.  This is the focus of Section \S\ref{sec:evaluation}.

  %  The rest of the paper is structured as follows. Section \ref{sec:related-work} summarizes the literature related to problems we care about. Section \ref{sec:system-overview} glances at the overall structure, assumptions and goals of the oracle system. Section \ref{sec:aced} expounds ACeD as scalable data availability oracle and analyzes its expected properties. More details about the system including the incentive design are elaborated in Section \ref{sec:incentive}. Section \ref{sec:implementation} presents the performance of the implementation ETH Stent and 
  We conclude the paper with an overview of our contributions in the context of the interoperability of blockchains in Section~\S\ref{sec:conclusion}.

\section{Related Work}
\label{sec:related-work}

\noindent {\bf Blockchain Scaling}. Achieving  the highest throughput and lowest latency blockchains has been a major focus area; this research has resulted in new proof of work \cite{bagaria2019prism,yang2019prism,yu2020ohie,fitzi2018parallel}, proof of stake \cite{gilad2017algorand,daian2019snow,kiayias2017ouroboros,david2018ouroboros,badertscher2018ouroboros}, and hybrid \cite{buterin2017casper,pass2017hybrid} consensus protocols. Scaling throughput linearly with the number of nodes in the network (known as {\em horizontal scaling}) is the focus of {\em sharding} designs: partition the blockchain and parallelize the computation and storage responsibilities  \cite{luu2016secure, kokoris2018omniledger,rana2020free2shard}. Both horizontal and vertical scaling approaches directly impact the core technology (`` layer 1") of the blockchain stack and their implementation in existing public blockchains is onerous. An alternate approach is to lock parts of the state of the blockchain and process transactions associated with this locked state in an application layer: this approach includes payment channels\cite{decker2015fast, miller2017sprites} and state channels \cite{poon2017plasma,teutsch2019scalable, kalodner2018arbitrum}.  %  Truebit and Arbitrum \cite{teutsch2019scalable, kalodner2018arbitrum} upload the program results along with proof on the main chain which acts as a dispute mediator and waits for a challenge to resolve the disagreement. 
In this paper we are concerned with a third (and direct) form of blockchain scaling:  we interact very lightly with the trusted blockchain (only storing hashes of each block of the side blockchains) but design and implement an oracle that allows for scalable secure interactions between the side blockchains and the trusted blockchain. 

\noindent {\bf Data Availability}. 
Blockchain nodes that do not have access to all of the data (e.g.: light nodes in simple payment verification clients in Bitcoin) are susceptible to the data availability attack. One approach makes use of zero knowledge proofS to represent the delta of the blockchain state, eliminating the need for the entire data\cite{zkrollup}. Another approach is to have the light nodes rely on full nodes to notify the misbehavior of a malicious block proposer (which has withheld the publication of the full block). Coding the blockchain data to improve the efficiency of fraud proofs was first suggested in \cite{al2018fraud} (using 2D Reed-Solomon codes), and was strongly generalized by \cite{yu2020coded} via a cryptographic hash accumulator called Coded Merkle Tree (CMT) to generate the commitment of the block. Light nodes randomly sample symbols through anonymous channels and by collecting sampling results, an honest full node can either reconstruct the block or provide incorrect-coding proof. CMT reduces the proof size compared to \cite{al2018fraud}. However, the sampling mechanism is probabilistic (and not appropriate for implementation in the oracle setting of this paper); further CMT requires anonymous communication channels. In this paper, we propose an alternate coded Merkle tree construction (CIT) that is specific to the oracle operation: the method is equipped with a ``push" scheme to efficiently and deterministically disseminate data among oracle nodes (no communication anonymity is needed). 

Another perspective to understand the data availability problem is to find a secure, efficient way of dispersing data in the network and prove it is permanently retrievable. Such a problem has been discussed under the context of asynchronous verifiable information dispersal (AVID)\cite{cachin2005asynchronous}, which applies erasure code to convert a file into a vector of symbols which later are hashed into a fingerprint. The symbols together with the fingerprint are disseminated to nodes. Then each node checks the hash and broadcasts its received copy for guaranteeing  retrievability. AVID improves storage saving but sacrifices communication efficiency; the performance is  summarized in Table~\ref{tab:comparison}, with details  in Appendix \ref{sec:performance-table}. The key insight of ACeD is that by using codes designed to have short fraud proofs, it is possible to avoid the additional echo phase in AVID where honest nodes flood the chunks to each other and thus provide near-optimal communication efficiency.

\section{System and Security Model}
\label{sec:system-overview}
The system is made up of three components: a trusted blockchain (that stores commitments and decides ordering), clients (nodes in side blockchains who propose data), 
and an intermediate oracle layer ensuring data availability (see Fig.\ref{fig:ACeD}).

\subsection{Network Model and Assumptions}
There are two types of nodes in the network:  oracle nodes and clients.%side nodes and oracle nodes. 

%\textit{Side-chain clients} are application-level users who send transactions and validate contract state changes based on committed data. Like light nodes in Bitcoin, they are not interested in storing and verifying the whole blockchain and only care about specific transactions within the scope of their applications.

%\textit{Side nodes} are running in side blockchains, different chains have full autonomy and perform functionality independently. Side nodes propose blocks within their own blockchain and meanwhile broadcast them to the oracle layer network. They periodically update the local ledger according to the commitment states from the trusted blockchain and ask oracle nodes for the missing blocks. 

\textit{Oracle nodes} are participants in the oracle layer. They receive block commitment requests from clients, including block headers, and a set of data chunks. %when a new block is dispersed by any side node. 
After verifying the  integrity and correctness of the data, they vote to decide whether the block is available or not and submit the results to the trusted blockchain. %On the purpose of efficiency, a \textit{committee} is randomly selected among all oracle nodes to collect votes and submit results to the trusted blockchain. 
%\peiyao{*remove introducing committee since it's protocol details instead of model components.}

\textit{Clients} propose blocks and request the oracle layer to store and commit the blocks. They periodically update the local ledger according to the commitment states from the trusted blockchain and ask oracle nodes for the missing blocks on demand. 

One of the key assumptions of our system is that the trusted blockchain has a persistent order of data and liveness for its service. Besides, we assume that in the oracle layer, there is a majority of honest nodes. %, and the adversary is static, which means corrupted nodes do not change after the protocol starts. For each side blockchain, we assume at least one side node is honest. 
For clients, we only assume that at least one client is honest (for liveness). Oracle nodes are connected to all clients. The network is synchronous, and the communication is authenticated and reliable.%, clients are connected to other clients within the same side blockchain.  %The system provides censorship resistance for all parties to access various information like trusted blockchain contract states. The commitment on the trusted blockchain for each side blockchain is maintained with a single accumulated hash which implies that there is no fork for side blockchains.

%\textit{Consensus nodes} are nodes working for the trusted blockchain. They do not care about the application execution since all the application-level data from different side blockchains will be serialized and homogenized when entering the view of the trusted blockchain. They participate in the trusted blockchain consensus protocol to decide which block to be added to the trusted blockchain. 

\subsection{Oracle Model}
The oracle layer is introduced to offload the storage and ensure data availability. The network of oracle layer consists of $N$ oracle nodes, which can interact with clients to provide data availability service. There exists an adversary that is able to corrupt up to $\beta N$ oracle nodes. Any node if not corrupted is called honest. 

The basic data unit for the oracle layer is a \textit{block}. A data availability oracle comprises of the following primitives which are required for committing and retrieving a block $B$.
\begin{enumerate}
    %\item \textbf{Accept Request:} A client delivers a request $(\textsf{request}, B, c)$ specifying a block $B$ with commitment $c$ to propose.
    \item \textbf{Generate chunks}: When a client wants to commit a block $B$ to the trusted blockchain, it runs $(\textsf{generate\_commitment}(B,M))$ to generate a commitment $c$ for the block $B$ and a set of $M$ chunks $c_1,..c_M$ which the block can be reconstructed from.
    \item \textbf{Disperse chunks}: There is a dispersal protocol $\textsf{disperse}(B,(c_1,..,c_M),N)$ which can be run by the client and specifies which chunks need to be sent to which of the $N$ oracle nodes. 
    \item \textbf{Oracle finalization}: The oracle nodes run a finalization protocol to finalize and accept certain blocks whose commitments are written into the trusted blockchain.
    %This primitive can be run on the client side 
    %\item \textbf{Disperse Data:} The block $B$ is disseminated to oracle nodes, the data chunk stored by node $i$ is dispersed by $\textsf{disperse}(B, c, i)$.
    \item \textbf{Retrieve Data:} Clients can initiate a request $(\textsf{retrieve}, c)$ for retrieving a set of chunks for any commitment $c$ that has been accepted by the oracle.
    \item \textbf{Decode Data}: There is a primitive $\textsf{decode}(c,\{c_i\}_{i \in S})$ that any client can run to decode the block from the set of chunks $\{c_i\}_{i \in S}$ retrieved for the commitment. The decoder also returns a proof that the decoded block $B$ is related to the commitment. %, or if the decoded block is $\emptyset$ that the commitment proposer was malicious. 
    %collect information from oracle nodes until the block is reconstructable, or a incorrect-coding proof is generated.
\end{enumerate}
%Notice that once a block is dispersed through a DA oracle, it is guaranteed to have persistence and can be retrieved any time on demand, which is known as the key property, \textbf{Availability}. As an information dispersal protocol, other expected properties are also satisfied and defined below. Furthermore, we expect a new property, \textbf{Efficiency} for the oracle when generating incorrect-coding proof. 
We characterize the security of the oracle model and formally define data availability oracle as follows,
\begin{definition} A data availability oracle for a trusted blockchain accepts blocks from clients and writes commitments into the trusted blockchain with the following properties:
\begin{enumerate}
    \item \textbf{Termination:} If an honest client initiates a $\textsf{disperse}$ request for a block $B$, then block $B$ will be eventually accepted and the commitment $c$ will be written into the trusted blockchain.
    %\item \textbf{Agreement:} Whenever an honest node completes the dispersal, all the other honest nodes eventually complete the dispersal.
    \item \textbf{Availability} If a dispersal is accepted, whenever an honest client requests for retrieval of a commitment $c$, the oracle is able to deliver either a block $B$ or a null block $\emptyset$ and prove its relationship to the commitment $c$.
    \item \textbf{Correctness:} If two honest clients on running  $(\textsf{retrieve}, c)$ receives $B_1$ and $B_2$, then $B_1=B_2$.  If the client that initiated the dispersal was honest, we require furthermore that $B_1=B$, the original dispersed block.
%    \item \textbf{Efficient Proof:} If the retrieved file for a commitment $c$ is a null file $\emptyset$, then the decoder should be able to give a {\em short} proof that the disperser of $c$ was malicious.
    %whenever an honest client sends request $(\textsf{retrieve}, c)$ , it eventually reconstructs $B'$ for commitment $c$, then $B = B'$.
    %\item \textbf{Efficiency:}
    %If the retrieval answer is $\bot$ then there should be an efficient way of proving that the retrieval is $\bot$ without having to download all the forwarded chunks.   

\end{enumerate}
\label{def:daoracle}
\end{definition}

A naive oracle satisfying all above expectations is trivial to construct, e.g., sending all oracle nodes a full copy of data. However, what we want is a \textbf{scalable} data availability oracle. To better understand this motivation, in the next section, we will introduce some metrics to concretize the oracle properties.

\section{Technical Description of ACeD}
\label{sec:aced}
In this section, we describe the four components of ACeD: CIT, dispersal protocol, retrieval protocol and block peeling decoder. 

%\vspace{-5mm}

\begin{figure}[hbt]
    \centering
    \includegraphics[width=0.9\textwidth]{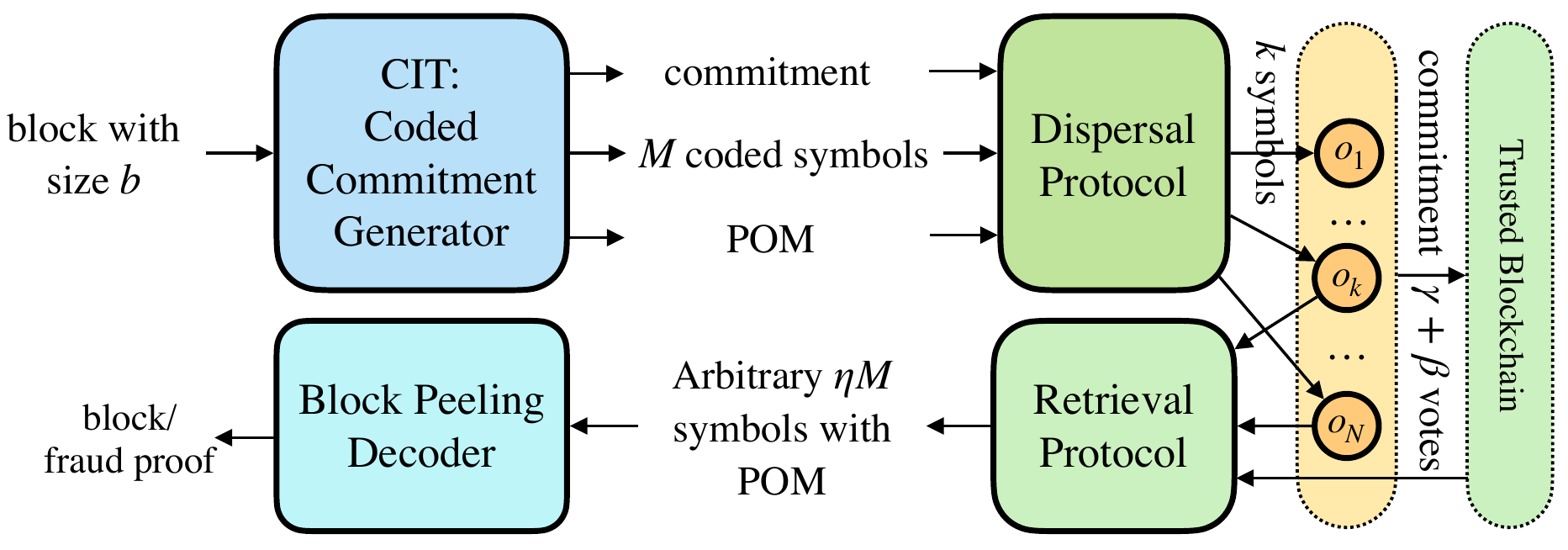}
    \caption{The pipeline for a block to be committed in trusted blockchain.}
    \label{fig:pipeline}
\end{figure}

%\vspace{-8mm}

%\noindent {\bf Coded Interleaving Tree}.  
\subsection{Coded Interleaving Tree}
The first building block of ACeD is a coded commitment generator which takes a block proposed by a client as an input and creates three outputs: a commitment, a sequence of {\em coded symbols}, and their proof of membership \POM -- see Figure \ref{fig:pipeline}. The commitment is the root of a coded interleaving tree (CIT), the coded symbols are the leaves of CIT in the base layer, and \POM for a symbol includes the Merkle proof (all siblings' hashes from each layer) and a set of {\em parity symbols} from all intermediate layers. A brief overview to erasure coding can be found in Appendix~\ref{sec:erasure-code}. 

The construction process of an example CIT is illustrated in figure \ref{fig:tree}. Suppose a block has size $b$, and its CIT has $\ell$ layers. The first step to construct the CIT is to divide the block evenly into small chunks, each is called a {\em systematic symbol}. The size of a systematic symbol is denoted as $c$, so there are $s_{\ell} = b/c$ systematic symbols. And we apply erasure codes with coding ratio $r \le 1$ to generate $m_{\ell} =s_\ell/r$  coded symbols in base layer. Then by aggregating the hashes of every $q$ coded symbols we get $m_{\ell}/q$ systematic symbols for its parent layer (layer $\ell-1$), which can be further encoded to $m_{\ell-1} = m_\ell/(qr)$ coded symbols. We aggregate and code the symbols iteratively until the number of symbols in a layer decays to $t$, which is the size of the root.

\begin{figure}[hbt]
    \centering
    \includegraphics[width=\linewidth]{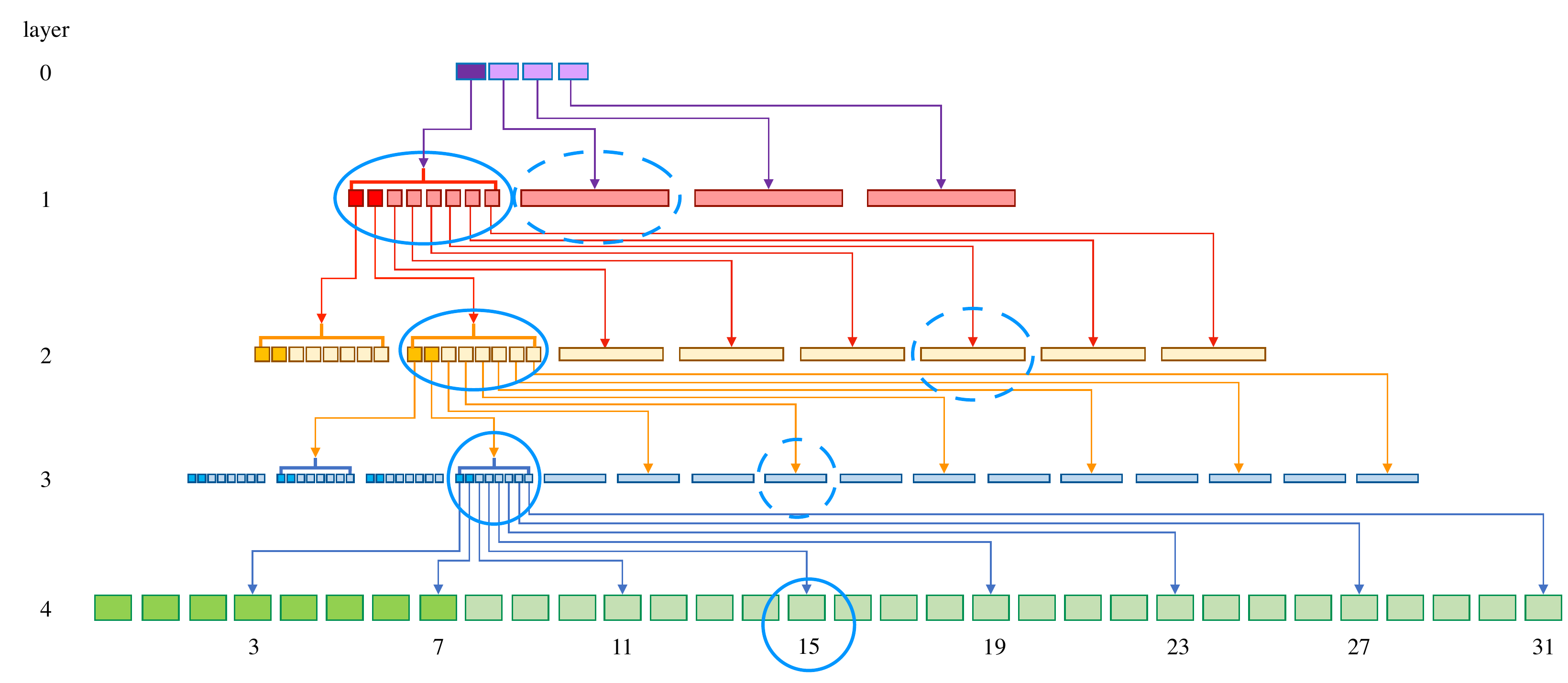}
    \caption{(1) CIT construction process of a block with $s_\ell = 8$ systematic symbols, applied with erasure codes of coding ratio $r = \frac{1}{4}$. The batch size $q = 8$ and the number of hashes in root is $t=4$. (2) Circled symbols constitute the 15th base layer coded symbol and its \POM. The solidly circled symbols are the base layer coded symbol and its Merkle proof (intermediate systematic symbols), the symbols circled in dash are parity symbols sampled deterministicly.}
    \label{fig:tree}
\end{figure}

For all layers $j$ except for root, $1\le j \le \ell$, denote the set of all $m_j$ coded symbols as $M_j$, which contains two disjoint subsets of symbols: systematic symbols $S_j$ and parity symbols $P_j$. The number of systematic symbols is $s_j = r m_j$. Specifically, we set $S_j = [0, rm_j)$ and $P_j=[rm_j, m_j)$. Given a block of $s_{\ell}$ systematic symbols in the base layer, the aggregation rule for the $k$-th systematic symbol in layer $j-1$ is defined as follows:
\begin{equation}
  Q_{j-1}[k] = \{ \h(M_{j}[x])\ |\ x \in [0,M_{j}), k = x \bmod rm_{j-1} \} 
\end{equation}
\begin{equation}
M_{j-1}[k] = \h(\mathsf{concat}(Q_{j-1}[k]))
\end{equation}
where $1\le j \le \ell$ and $\h$ is a hash function. $Q[k]$ is the tuple of hashes that will be used to generate $k$-th symbol in the parent layer and \textsf{concat} represents the string concatenation function which will concatenate all elements in an input tuple.

Generating a \POM for a base layer symbol can be considered as a layer by layer sampling process as captured by the following functions:
\begin{equation*}
	f_{\ell}: [m_{\ell}] \to {m_{\ell-1} \choose 2}, \cdots,  f_{2}: [m_{\ell}] \to {m_{1} \choose 2};
\end{equation*}
Each function maps a base layer symbol to two symbols of the specified layer: one is a systematic symbol and the other is a parity symbol. %The key property is that these two symbols are {\em siblings}  so  no extra Merkle proof is needed (see Lemma~\ref{lm:sibl} for a formal statement). 
We denote the two symbols with a tuple of functions $f_j(i) = (p_j(i), e_j(i))$, where $p_j(i)$ is the sampled systematic symbol and $e_j(i)$ is the sampled parity symbol, each is defined as follows:

\begin{equation}
\label{eq:pe}
    p_{j}(i) = i \bmod rm_{j-1}  ; \;\; e_{j}(i) = rm_{j-1}  + (i \bmod (1-r) m_{j-1})
\end{equation}
where $p_{j}: [m_{\ell}] \to [0, rm_{j-1})$ and $e_{j}: [m_{\ell}] \to [rm_{j-1}, m_{j-1})$. Note that the sampled non-base layer systematic symbols are automatically the Merkle proof for both systematic and parity symbols in the lower layer.

There are two important properties of the sampling function. (1) It guarantees that if at least $\eta\le 1$ ratio of distinct base layer symbols along with their \POM are sampled, then in every intermediate layer, at least $\eta$ ratio of distinct symbols are already picked out by the collected \POM. It ensures the reconstruction of each layer of CIT(Lemma~\ref{lm:recon}, {\em reconstruction} property). (2) All sampled symbols at each layer have a common parent (Lemma~\ref{lm:sibl}, {\em sibling} property), which ensures the space efficiency of sampling.

As described above,  CIT can be represented by parameters \pt.  All parameters have been defined except  $\alpha$, which  is the {\em undecodable ratio} of an erasure code and $d$, which is the number of symbols in a failed {\em parity equation}. Both of them will be discussed in the decoder section (Section~\ref{sec:decoder}). 
%After receiving dispersed data chunks, oracle nodes will first check the validity using a \POM Verifier. Specifically, a commitment can be computed using any stored symbol and its Merkle proof, if there exists a calculated commitment $c'$ inconsistent with the stored commitment $c$, the verification is recognized as failed. Notice that oracle nodes are unable to check the correctness of coding at this time, if the \POM verification passes they then generate BLS signatures\cite{bls} using registered key sets and multicast the signatures to a selected committee as votes. Committee members are elected randomly in each round, they wait for signatures from oracle nodes until reaching a threshold $\theta$. If a committee member received more than $\theta N$ valid signatures, it will aggregate all the received votes and commit the aggregated signature to the trusted blockchain. Then it is the smart contract's responsibility to verify aggregated signatures, check signers' identity and decide whether the block is eligible to be added. A valid signature indicates the success of dispersal otherwise oracle nodes will drop the stored data to save the space when timeout happens.

\myparagraph{Comparison with CMT}  (1) CIT is designed for the Push model (the client sends different chunks to different nodes) whereas CMT is designed for the Pull model (the nodes decide which chunks to sample from the client). (2)The CIT provides an ideal property (Lemma 2) that ensures reconstruction from any set of nodes whose size is greater than a given threshold. Thus as long as enough signatures are collected and we can be assured that the honest subset represented therein is larger than this threshold, we are guaranteed reconstruction in CIT. This is not the case in CMT, whose guarantees are probabilistic since each node samples data in a probabilistic manner. (3) On the technical side, CIT requires a new interleaving layer as well as a dispersal mechanism which decides how to assign the different symbols from intermediate layers to different groups. (4) Furthermore, CMT requires an assumption of an anonymous network whereas CIT does not require an anonymous network. 

\subsection{Dispersal Protocol}
\label{section:dispersal}
%When deciding which portion of data a node needs to store, we consider both availability and efficiency, that is, how to make sure a block can be reconstructed with data collected from arbitrary subset of nodes given a threshold while saving as much space as possible for oracle nodes. 
Given the commitment, coded symbols and \POM generated by CIT, the dispersal protocol is used to decide the chunks all oracle nodes need to store. Consider a simple model where there are $M = b/(cr)$ chunks to be distributed to $N$ nodes (each node may receive $k$ chunks) such that any $\gamma$ fraction of nodes together contains $\eta$ fraction of chunks. The efficiency of the dispersal protocol is given by $\lambda = M/(Nk)$. We define the dispersal algorithm formally as follows,

\begin{definition}
The chunk dispersal algorithm is said to achieve parameters set $(M,N,\gamma,\eta,\lambda)$ if there is a collection of sets $\mathcal{C} = \{A_1,...,A_N\}$ such that  $A_i \subseteq [M]$, $|A_i| = \frac{M}{N \lambda}$. Also, for any $S \subseteq [N]$ with $|S| = \gamma N$, it holds that $|\bigcup\limits_{i \in S}{} A_i| \geq \eta M$.
\end{definition}

To simplify the $5$ dimensional region, we consider the tradeoff on the three quantities: a certain triple $(\gamma,\eta,\lambda)$ is said to be achievable if for any $N$, there exists $M$ such that the chunk dispersal algorithm can achieve parameters set $(M,N,\gamma,\eta,\lambda)$. In the ACeD model, $\gamma$ is constrained by security threshold $\beta$ since the clients can only expect to collect chunks from at most $(1-\beta) N$ nodes. $\eta$ is constrained by the undecodable ratio $\alpha$ of the erasure code since the total chunks a client collects should enable the reconstruction. So for a given erasure code, there is a trade-off between dispersal efficiency and security threshold. 

Our main result is the demonstration of a dispersal protocol with near optimal parameters (Lemma~\ref{thm:dispersal}). 

\begin{lemma}
\label{thm:dispersal}
    If $\frac{\gamma}{\lambda} < \eta$, then $(\gamma,\eta,\lambda)$ is not feasible. 
    If $\frac{\gamma}{\lambda} > \log (\frac{1}{1-\eta})$ then $(\gamma,\eta,\lambda)$ is feasible and there exists a chunk dispersal protocol with these parameters. 
\end{lemma}

We provide a sketch of proof. If $\frac{\gamma}{\lambda} < \eta$, the maximum number of distinct symbols is $\frac{M}{N\lambda}\cdot \gamma M = \frac{M\gamma}{\lambda} < \eta M$, so  $(\gamma,\eta,\lambda)$ is not feasible. 

If $\frac{\gamma}{\lambda} > \log (\frac{1}{1-\eta})$, we prove that the probability of $\mathcal{C}$ is not a valid code can vanish to as small as we want. The problem is transformed to the upper bound of $P(Y < \eta M)$ use a standard inequality from the method of types \cite{csiszar1998method}, where $Y$ is the number of distinct chunks sampled. By sampling $\frac{\gamma}{\lambda} M$ chunks from a set of $M$ chunks randomly, we have $P(Y < \eta M) \leq e^{-f(\cdot) \cdot M } $, where $f(\cdot)$ is a positive function. So the probability of $\mathcal{C}$ is a valid code can be positive, which proves the existence of a {\em deterministic} chunk distribution algorithm. And we can control the probability of $\mathcal{C}$ is not a valid code to be vanishingly small. The detailed proof of the result can be found in Appendix~\S\ref{sec:proofdispersal}. 
%The results indicate that as long as the probability of being a valid code is strictly greater than zero, then it proves the existence of a {\em deterministic} chunk distribution algorithm for feasible dispersal parameters and thus guarantees the optimal adversary tolerance.

%Our main result will be that as long as $\gamma > 2.08 \eta \lambda$ we show that a chunk dispersal protocol that achieves these parameters.

%\subsubsection{Chunk dispersal protocols from Steiner Systems}

%We will show that ideas from combinatorial design theory can be used to achieve certain parameters $(\gamma,\eta,\lambda)$. 

In the dispersal phase, all oracle nodes wait for a data commitment $c$, $k$ assigned chunks and the corresponding \POM. The dispersal is accepted if $\gamma + \beta$ fraction of nodes vote that they receive the valid data.

%When deciding which portion of data a node needs to store, we consider both availability and efficiency, that is, how to make sure a block can be reconstructed with data collected from arbitrary subset of nodes given a threshold while saving as much space as possible for oracle nodes. Specifically, an oracle node needs to download a data commitment $c$ and $m/n$ units of dispersed data where $m$ is the total number of coded symbols and $n$ is the number of oracle nodes. %A unit of dispersed data is made up of two parts, a base layer symbol with Merkle proof and proof of membership (\POM). Similar to the classic Merkle tree, the Merkle proof of a coded symbol includes all siblings' hashes from each layer during the aggregating process. \POM  contains a set of systematic symbols and parity symbols to make sure that all the intermediate layers have enough data to reconstruct the whole layer using the peeling decoder. We define a set of functions to generate \POM layer by layer:

%BLS uses a bilinear pairing $e:\Gone\times \Gtwo\rightarrow \Gt$ between cyclic subgroups $\Gone$ and $\Gtwo$ of suitable elliptic curves points with values in a group of units $\Gt$. We also assume two hash functions $\hzero:\mathcal{M}\rightarrow \Gone$ which maps arbitrary messages to an element in $\Gone$ and $ \hone: \mathcal{M}\rightarrow \mathbb{Z}_q$ which maps arbitrary messages to an integer. The detailed algorithm to generate key and signatures are stated in Algorithm 2. 

\subsection{Retrieval Protocol and Block Decoding}
\label{sec:decoder}
When a client wants to retrieve the stored information, the retrieval protocol will ask the oracle layer for data chunks. Actually, given erasure codes with undecodable ratio $\alpha$, an arbitrary subset of codes with the size of over ratio $1-\alpha$ is sufficient to reconstruct the whole block. When enough responses are collected, a hash-aware peeling decoder introduced in \cite{yu2020coded} will be used to reconstruct the block. The decoding starts from the root of CIT to the leaf layer and for each layer, it keeps checking all degree-0 parity equations and then finding a degree-1 parity equation to solve in turn. Eventually, either all symbols are decoded or there exists an invalid parity equation. 
In the second case, a logarithmic size incorrect-coding proof is prepared, which contains the inconsistent hash, $d$ coded symbols in the parity equation and their Merkle proofs. After an agreement is reached on the oracle layer, the logarithmic size proof is stored in the trusted blockchain to permanently record the invalid block. Oracle nodes then remove all invalid symbols to provide space for new symbols. In the special case when the erasure code used by the CIT requires more than $(1-\alpha)$ ratio of symbols to decode the original block, oracle nodes need to invoke a bad code handling protocol to reset a better code. We leave details in the Appendix~\S\ref{sec:bad-code}.

%\alter{
%\subsection{Client actions on invalid encoded symbols}
%    Side nodes rely on Oracle nodes for data retrieval, and reply fraud proof when reconstruction error is encountered. Because of the voting process, each side node can collect sufficient symbols for starting reconstruction whose outcome depends on whether there is any invalid coded symbol. When a symbol is decoded inconsistently to its hash revealed from its upper layer, a logarithmic size fraud proof is prepared for broadcasting to all Oracle nodes; the proof itself contains the inconsistent hash, constant number of symbols in the parity equation and their merkle proofs. After an agreement is reached on the Oralce layer, the logarithmic size proof is stored in the trusted blockchain to perminantly record the invalid block, as well as to punish the block producer. Oracle nodes then remove all invalid symbols to provide space to new symbols. 
%}

\subsection{Protocol Summary}
In summary, an ACeD system with $N$ oracle nodes and block size $b$ using CIT $\mathcal{T}$ and dispersal protocol $\mathcal{D}$ can be represented by parameters \paced, where \pt and \pd. The pipeline to commit a block to the trusted blockchain is as follows (see Figure~\ref{fig:pipeline}).

\begin{itemize}
    \item A client proposes a block of size $b$, it first generates a CIT with base layer symbol size $c$, number of hashes in root $t$, coding ratio $r$ and batch size $q$. There are $M = b/(cr)$ coded symbols in the base layer. And then it disperses $M$ coded symbols, their \POM and the root of CIT to $N$ oracle nodes using the dispersal protocol \pd. 
    \item Oracle nodes receive the dispersal request, they accept chunks and commitment, verify the consistency of data, \POM and root, vote their results. A block is successfully committed if there are at least $\beta + \gamma$ votes. Upon receiving retrieval requests, oracle nodes send the stored data to the requester. Upon receiving the fraud proof of a block, oracle nodes delete the stored data for that block.
    %\item Oracle committee collects votes, if enough signatures are aggregated, it submits the commitment and aggregated signature to the trusted blockchain. 
    \item Other clients %who are in the same side blockchain as the block proposer accept forward blocks. If they find that commitment on the trusted blockchain has been updated but no corresponding block is received, they 
    send retrieval requests to the oracle nodes on demand. Upon receiving at least $\eta \ge 1-\alpha$ fraction of chunks from at least $\gamma$ oracle nodes, they reconstruct the block, if a coding error happens, the fraud proof will be sent to the trusted blockchain. 
\end{itemize}

\section{Performance 
Guarantees of ACeD}
\label{sec:analysis}

\begin{theorem}
\label{theorem-main}
Given an adversarial fraction $\beta < \frac{1}{2}$ for an oracle layer of $N$ nodes, ACeD is a data availability oracle for a trusted blockchain with $O(b)$ communication complexity, $O(1)$ storage and download overhead in the normal case, and $O(\log b)$ storage and download overhead in the worst case.
\end{theorem}

This result follows as a special case of a more general result below (Theorem~\ref{theorem-detail}).  
\begin{proof}
Suppose $\chi$ is an ACeD data availability oracle with parameters \paced \ where \pt and \pd. There are at most  $\beta < \frac{1}{2}$ fraction of adversarial nodes in the oracle layer. Then by setting $r,q,d,t=O(1),c=O(\log b), b \gg N$,  $\chi$ is secure as long as $\beta \le \frac{1}{2}(1-\lambda\log(\frac{1}{\alpha}))$ %$\frac{1}{\lambda} > \frac{\log{(\alpha^{-1})}}{1-2\beta}$
; the communication complexity of $\chi$ is $O(b)$ because
\[Nyt + \frac{b}{\lambda r} + \frac{(2q-1)by}{cr \lambda}\log_{qr}\frac{b}{ctr} = O(N)  + O(b) + O(b) = O(b)\]
the storage and download overhead in the normal case is $O(1)$, because
\[\frac{Nyt}{b} +\frac{1}{\lambda r} + \frac{(2q-1)y}{cr\lambda}\log_{qr}\frac{b}{ctr} = O(1) + O(1) + O(\frac{1}{\log b}\log(\frac{b}{\log b})) = O(1)\]
the storage and download overhead in the worst case is $O(\log b)$, because
\[\frac{c(d-1)}{y} + d(q-1)\log_{qr}\frac{b}{ctr} = O(\log b) + O(\log_{qr}(\frac{b}{\log{b}}) = O(\log b) ). \]
\end{proof}

A complete description of the security and performance guarantees of ACeD is below.  
\begin{theorem}
\label{theorem-detail}
ACeD is a data availability oracle for the trusted blockchain tolerating at most $\beta \le 1/2$  fraction of adversarial oracle nodes in an oracle layer of $N$ nodes. The ACeD is characterized by the system parameters \paced, where \pt and \pd. $y$ is a constant size of a hash digest, then 
\begin{enumerate}
\item ACeD is secure under the conditions that 
    \[\beta \le \frac{1-\gamma}{2};\ \frac{\gamma}{\lambda} > \log (\frac{1}{1-\eta});\  \eta \ge 1-\alpha\]
    \item Communication complexity is 
    \[Nyt + \frac{b}{\lambda r} + \frac{(2q-1)by}{cr \lambda}\log_{qr}\frac{b}{ctr}\]
    \item In normal case, both the storage and download overhead are 
        \[\frac{Nyt}{b} +\frac{1}{\lambda r} + \frac{(2q-1)y}{cr\lambda}\log_{qr}\frac{b}{ctr}\]
    \item In worst case, both storage and download overhead are
        \[\frac{c(d-1)}{y} + d(q-1)\log_{qr}\frac{b}{ctr}\]
\end{enumerate}
\end{theorem}

\begin{proof}
We prove the security and efficiency guarantees separately.

\subsection{Security}
\label{sec:securityanalysis}
To prove that ACeD is secure as long as the trusted blockchain is persistent and \[1 - 2\beta \ge \gamma;\ \frac{\gamma}{\lambda} > \log (\frac{1}{1-\eta});\  \eta \ge 1-\alpha\] , we prove the following properties as per Definition~\ref{def:daoracle}.
\begin{itemize}
    \item \textbf{Termination.} In ACeD, a dispersal is accepted only if there is a valid commitment submitted to the trusted blockchain. Suppose an honest client requests for dispersal but the commitment is not written into the trusted blockchain, then either the commitment is not submitted or the trusted blockchain is not accepting new transactions. Since $1-2\beta \ge \gamma$, thus $\beta + \gamma \le 1-\beta$, even if all corrupted nodes remain silent, there are still enough oracle nodes to vote that the data is available and the commitment will be submitted, hence the trusted blockchain is not live, which contradicts our assumption.
    
% One potential fault is that when the committee is occupied by adversaries who stay mute to violate the liveness property. The following lemma shows that such a case will not happen except with negligible probability (see the proof in Appendix %\S\ref{proof-lm-committee}).
%    \begin{lemma}
%\label{lm:committee}
%There exists at least one honest node in uniformly selected Committee except with negligible probability.
%\end{lemma}
%     Then the proof of termination is trivial. If an honest client requests to disperse a block, it follows the protocol to construct CIT and send different portions of data to corresponding nodes. All honest nodes will send votes to committee members upon receiving valid dispersal units. According to lemma \ref{lm:committee}, the honest committee member will submit the aggregated signature given chosen signature aggregation parameter $\theta\le 1-\beta$. Thus the block is dispersed successfully.
    \item \textbf{Availability.} If a dispersal is accepted, the commitment is on the trusted blockchain and $\beta+\gamma$ oracle nodes have voted for the block. Since the trusted blockchain is persistent, whenever a client wants to retrieve the block, it can get the commitment and at least $\gamma$ nodes will respond with the stored chunks. On receiving chunks from $\gamma$ fraction of nodes, for a CIT applying an erasure code with undecodable ratio $\alpha$ and a feasible dispersal algorithm $(\gamma, \eta, \lambda)$(Lemma \ref{thm:dispersal}), because $\eta \ge 1-\alpha$, the base layer is reconstructable. Then we prove the following lemma ensures the reconstruction of all intermediate layers.
    \begin{lemma} 
\label{lm:recon}
(Reconstruction) For any subset of base layer symbols $W_\ell$, denote $W_j := \bigcup_{i\in W_\ell}f_j(i)$ as the set of symbols contained in \POM of all symbols in $W_\ell$. If $|W_\ell|\ge \eta m_\ell$, then $\forall j\in[1,\ell]$, $|W_j|\ge \eta m_j$.
\end{lemma}
     The proof of Lemma~\ref{lm:recon} utilizes the property when generating \POM given base layer symbols. (See details in Appendix \S\ref{proof-lm-recon}). Thus the entire tree is eventually reconstructed and the oracle can deliver a block $B$, and the proof for $B$'s relationship to commitment $c$ is the Merkle proof in CIT. If a client detects a failed parity equation and outputs a null block $\emptyset$, it will generate an incorrect-coding proof. 
    
    \item \textbf{Correctness.}
    Suppose for a given commitment $c$, two honest clients reconstruct two different blocks $B_1$ and $B_2$, the original dispersed block is $B$.
    \begin{enumerate}
        \item[(1)] If the client that initiated the dispersal was honest, according to the availability property, $B_1,B_2 \ne \emptyset$, both clients can reconstruct the entire CIT. If $B_1 \ne B_2$, the commitment $c_1 \ne c_2$, which contradicts our assumption that the trusted blockchain is persistent.
        \item[(2)]If the client that initiated the dispersal was adversary and one of the reconstructed blocks is empty, w.l.o.g suppose $B_1 = \emptyset$, the client can generate a fraud proof for the block. If $B_2 \ne \emptyset$, the entire CIT is reconstructed whose root is commitment $c_2$. Since there is no failed equation in the CIT of $B_2$, $c_1 \ne c_2$, which contradict our assumption that the trusted blockchain is persistent. 
        \item[(3)]If the client that initiated the dispersal was adversary and $B_1,B_2 \ne \emptyset$, both clients can reconstruct the entire CIT. If $B_1 \ne B_2$, the commitment $c_1 \ne c_2$, which contradict our assumption that the trusted blockchain is persistent.
    \end{enumerate}
    Thus we have $B_1 = B_2$, and if the client that initiated the dispersal is honest, $B_1 = B$.
    
    %Using an erasure code whose undecodable ratio is at least $1-\eta$ for every layer, according to the availability property, we can ensure that if there are $\eta m_\ell$ coded symbols available in the base layer, then collected data would allow the recovery of all data symbols. The correctness can be proved layer by layer. Starting from the root $c$, once the peeling decoder decodes its child layer, the correctness for decoded symbols can be ensured by the given root, which is also the hashes of all symbols. Then the newly recovered layer will be responsible for checking the correctness of the next layer, until the base layer is decoded as $B'$ or an incorrect-coding proof is generated (the fraud proof indicates disperser is malicious). Since $c$ is a hash commitment of the original block $B$, denote the aforementioned coded commitment generator as $\textsf{G}$, we have $\textsf{G}(B) = \textsf{G}(B') $, then $B = B'$ except for a negligible probability.
    %\item \textbf{Efficiency.}
  %  We use incorrect-coding proof to let oracle nodes avoid downloading all forwarded data chunks when checking validity. To calculate the size of proof for $d$ coded symbols, it should consist of $d-1$   symbols and their Merkle proofs to represent a failed parity equation, thus the size is
%\begin{equation}
%    \frac{(d-1)b}{s_\ell} + dy(q-1)\log_{qr}\frac{s_\ell}{tr} = O(\log s_{\ell}) = O(\log b)
%\end{equation}
\end{itemize}

%\bx{ {\em Persistence as long as $\beta < 1/2$, $\gamma < 1 - \beta$ and the trusted blockchain is persistent}. Suppose not that there are at least two honest ACeD clients who cannot come into agreement about the block ledger, then the two clients must have two different ledgers. Because ACeD always guarantees data availability as long as $\beta < 1/2$ and $\gamma < 1 - \beta$, two clients with the identical block digest chain must be able to reconstruct all block in the same way. Therefore two clients must have received two different chain of hash digest from the trusted blockchain, which contradict our assumption that the trusted blockchain is persistent.

%{\em Liveness as long as $\beta < 1/2$, $\gamma < 1 - \beta $ and the trusted blockchain is live}. Suppose that no  transaction for honest clients can be included in the ledger, then either ACeD layer prevents the transaction from entering the system or the trusted blockchain is not live. Because $\gamma < 1 - \beta$ and $\beta < 1/2$, the honest nodes in Oracle layer can always collect sufficient signature and commit the block digest, hence the trusted blockchain is not live, which contradicts our assumption.

%We now prove that ACeD is a data availability oracle with any trusted blockchain.}

\subsection{Efficiency}
\label{sec:efficiencyanalysis}
Prior to computing the storage and communication cost for a single node to participate dispersal, we first claim a crucial lemma: 
\begin{lemma}
\label{lm:sibl}
For any functions  $p_j(i)$ and $e_j(i)$ defined in equation \ref{eq:pe}, where $1\le j\le \ell$, $0\le i< m_{\ell}$, $p_j(i)$ and $e_j(i)$ are siblings.
\end{lemma}
Lemma \ref{lm:sibl} indicates that in each layer, there are exactly two symbols included in the \POM for a base layer symbol and no extra proof is needed since they are siblings (see proof details in \S\ref{proof-lm-sibl}). For any block $B$, oracle nodes need to store two parts of data, the hash commitment, which consists of $t$ hashes in the CIT root, and $k$ dispersal units where each unit contains one base layer symbol and two symbols per intermediate layer. Denote the total storage cost as $X$, we have
\begin{align}
      X =  ty + kc + k[y(q-1) + yq]\log_{qr}\frac{b}{ctr}
\end{align}
where $y$ is the size of hash, $b$ is the size of block, $q$ is batch size, $r$ is coding rate, and $c$ is the size of a base layer symbol. Notice that $k = \frac{b}{Nrc \lambda}$, we have
\begin{align}
     X = ty + \frac{b}{Nr\lambda } + \frac{(2q-1)by}{Nrc \lambda}\log_{qr}\frac{b}{ctr}.
\end{align} 
It follows that the communication complexity is $NX$. In the normal case, each node only stores $X$ bits hence the storage overhead becomes $\frac{NX}{b}$, and similarly when a client downloads data from $N$ nodes, its overhead is $\frac{NX}{b}$.  In the worst case, we use incorrect-coding proof to notify all oracle nodes. The proof for a failed parity equation which contains $d$ coded symbols consist of $d-1$   symbols and their Merkle proofs, denote the size as $P$, we have
\begin{equation}
  P =  (d-1)c + dy(q-1)\log_{qr}\frac{b}{ctr}.
\end{equation}
\end{proof}
The storage and download overhead in this case is $\frac{P}{y}$, the ratio of the proof size and the size of reconstructed data ,a single hash $y$.

\section{Algorithm to System Design and Implementation}
\label{sec:implementation}

ACeD clients are nodes associated with  a number of side blockchains; the honest clients rely on ACeD and the trusted blockchain to provide an ordering service of their ledger (regardless of any adversarial fraction among the peers in the side blockchain). A client proposes a new block to all peers in the side blockchain by running ACeD protocol. An honest client confirms to append the new block to the local side blockchain once the block hash is committed in the trusted blockchain and the full block is downloaded. As long as there is a {\em single honest} client in the side blockchain, we have the following claim:
\begin{claim}
Once a block is confirmed by an honest client, security is guaranteed as long as the trusted blockchain is safe, even if the oracle layer is dishonest majority. Liveness is guaranteed when the trusted blockchain is live and the majority of the oracle layer is honest.
\end{claim}

The claim indicates that in side blockchains, the safety of a confirmed block only relies on the trusted blockchain because the commitment on it is irrefutable once the trusted blockchain is safe, and the honest client has already downloaded the full block. So even if the oracle layer is occupied by the dishonest majority, those confirmed blocks are still safe. However, the liveness relies on the liveness of both ACeD and the trusted blockchain. As for the side blockchain network, because data persistence is guaranteed by ACeD, any client inside the network can safely conduct a transaction and reconstruct the ledger without worrying about a dishonest majority; similarly a valid transaction will eventually join the ledger as long as there is a single honest client who can include the transaction into a block. 

Next we discuss the practical parameter settings for ACeD. We use these parameter choices to design and implement an ACeD oracle layer that interfaces with Ethereum as the trusted blockchain.

\noindent {\bf Parameter Specifications}. 
We study an ACeD system with $N=9000$ oracle nodes with adversarial fraction $\beta = 0.49$; the block size $b$ is $12$ MB and therefore $b \gg N$. In each block, the base layer symbol size $c$ is $2000 \log b \approx 48$ kB, which corresponds to $\frac{b}{c} \approx 256$ uncoded symbols in the base layer. Within the setup, we construct a CIT of five layers with parameters: number of root symbols $t=16$, hash size $y=32$ bytes, coding ratio $r=0.25$, aggregation batch size $q=8$ and $4$ erasure codes of size (256,128,64,32) for each non-root layer. For selecting erasure code properties, we use codes with undecodable ratio $\alpha=0.125$, maximal parity equation size $d=8$. In the dispersal protocol, we use $\eta = 0.875 = 1 - \alpha$, which translates to a dispersal efficiency $\lambda \le \frac{1-2\beta}{\log{\frac{1}{1-\eta}}} = 1/150$, and therefore each oracle node needs to store roughly 17 symbols. With ACeD characterized by those parameters, the total communication cost for a client to commit a $12$ MB block is roughly $5.38$ GB; this represents a $0.5N$ factor boost over storing just one block. In the normal path, after accepting the 12 MB block, each oracle node only has to store $448$ kB of data, a $3.7\%$ factor of the original data; if there is an incorrect-coding attack, oracle layer will together upload data of  size $339$ kB incorrect-coding proof to the trusted blockchain. To download a block in the normal path, a client can use a naive method of collecting all symbols from all oracle nodes. Despite the conservative approach at  block reconstruction, the entire download takes $5.38$ GB. A simple optimization involving selectively querying the missing symbols can  save significant download bandwidth:  a client only needs $896$ coded symbols in the base layer to reconstruct the block; thus, in the best case only $42$ MB is needed  to directly receive those symbols. When there is an incorrect-coding attack, a new user only needs to download the fraud proof which has been generated by other nodes (either client or oracle nodes); the proof is only of the size of $339$ kB.  Table~\ref{tab:comparison2} tabulates the performance metrics of ACeD (and baseline schemes) with these parameter settings. 

%\bx{Should I put the calculation below into the appendix \ref{sec:performance-table}?
In the uncoded repetition protocol, each node stores the identical block, so the storage cost for both cases is 12MB. In the uncoded dispersal protocol, because all oracle nodes are honest, the block is equally divided into $N$ nodes in both cases. In the AVID protocol, each message has the size $\frac{b}{N-2N\beta} + y(\log_2N+1)=4.37$ kB, so the total amount of messages exchanged during the broadcasting stage equals to 354GB; AVID does not commit an invalid block, so only the block hash is required to show the block is invalid. When $\beta=0.33$ 1D-RS uses the same erasure code as AVID to assign each node a symbol, so their normal case is identical; when there is an incorrect-coding attack, a node needs to download the whole block to reconstruct the block, therefore in the worst case both storage and download are 13.4MB. To tolerate $\beta=0.49$, 1D-RS needs to decrease the coding ratio to 0.02 while increasing the symbol size, because 200 nodes are required to start the decoding. A lower coding ratio requires nodes to store more $67.1$ kB data; the download cost are saved by reducing the number of Merkle proofs, because only 200 symbols are needed for decoding. Finally, we use the same method to evaluate ACeD with $\beta=0.33$, both normal case storage and communication complexity improves due to a better dispersal efficiency. In evaluating the metrics for 2D-RS, we first get a lower bound for the number of honest oracle nodes available in the decoding phase, which is used to compute the size of a matrix used in 2D-RS, and the goal is to decide a 2D-RS symbol size which guarantees each honest oracle node with at least one symbol. The 2D-RS looks over-perform ACeD when the adversarial ratio is 0.33. But if the block size increases from 12MB to 1.2GB while the ratio is 0.33, the worst case in ACeD needs 279kB, whereas 2D-RS becomes 21.8MB; when the adversary ratio is 0.49, ACeD needs 279 KB for a 1.2GB block, and 2D-RS requires a fault proof of 92.3MB. The downside of ACed is the communication complexity when the block size is large, a 1.2GB block requires each node communicating 64MB data, whereas 2D-RS only requires 6.6MB communication.

%the communication cost (according to theorem \ref{theorem-main})  is 121358 MB, which is $\lambda \frac{byq}{cr \lambda}\log_{qr}\frac{b}{tr} \approx 1e5 b $; in the normal case, the storage and download overhead is $121358$ bytes; in the worst case, the storage and download overhead is $938$ bytes.

\begin{table}[]
\caption{Performance metrics under a specific system parameterization.}
\centering
\begin{threeparttable}

\begin{tabular}{|l|l|c|c|c|c|c|}
\hline
\multirow{3}{*}{}    & maximal   & \multicolumn{2}{c|}{normal case }                              & \multicolumn{2}{c|}{worst case }                               & \multirow{3}{*}{\begin{tabular}[c]{@{}c@{}}communication \\ complexity\end{tabular}} \\ \cline{3-6}
                     & adversary & storage                       & download                      & storage                       & download                      &                                                                                      \\
                     & fraction  & \multicolumn{1}{l|}{cost\tnote{*}} & \multicolumn{1}{l|}{cost\tnote{*}  \tnote{$\dagger$}} & \multicolumn{1}{l|}{cost\tnote{*}} & \multicolumn{1}{l|}{cost\tnote{*}} &                                                                                      \\ \hline
uncoded (repetition) & 0.49       &     12MB               &         12MB                & 12MB                        & 12MB                        & 108GB                                                                              \\ \hline
uncoded (dispersal)  & 0       & 1.3kB                        & 12MB                        & 1.2kB                       & 12MB                       & 12MB                                                                               \\ \hline
AVID \cite{cachin2005asynchronous}                & 0.33       & 4.37kB                        & 13.4MB                        & 32B                        & 32B                        & 354GB                                                                              \\ \hline
1D-RS                & 0.33       &     4.37kB             &       13.4MB                     & 13.4MB                        &       13.4MB              & 39.4MB                                                                         \\ \hline
1D-RS                & 0.49       &     67.1kB             &       12.1MB                     & 12.1MB                        &       12.1MB              & 604MB                                                                         \\ \hline
2D-RS                & 0.33       &  5.4kB                &       16.6MB                     &      232.1KB                   & 232.1KB                    &           48.9MB                                                               \\ \hline
2D-RS                & 0.49       &   72.1kB               &     13MB                       & 925.6KB                        & 925.6KB                    &        648.6MB                                                              \\ \hline
ACeD                 & 0.33       &        50.3kB                 & 42MB                      & 339kB                   & 339kB                   & 452 MB                                                                          \\ \hline
ACeD                 & 0.49       &        597kB                 & 42MB                      & 339kB                   & 339kB                   & 5.38GB                                                                          \\ \hline
\end{tabular}
\begin{tablenotes}\footnotesize
%\item[*] AVID is an asynchronous protocol.
%\item See metrics definition in Table~\ref{tab:metrics}.
\item[*] cost is derived by $ \frac{b }{N} \cdot \text{overhead}$
\item[$\dagger$] best case
%\item all calculation details other than ACeD is presented in Appendix \ref{sec:performance-table}
\end{tablenotes}

\end{threeparttable}

\label{tab:comparison2}
\end{table}

The ACeD oracle layer interacts with the side blockchains and the trusted blockchain.  For a high performance (high throughput, low gas), the oracle layer and its interaction protocols have to be  carefully designed for software implementation.  In this section, we discuss our system design decisions that guided the implementation. 

\noindent {\bf Architecture}. 
We use Ethereum as the trusted blockchain. The  oracle nodes interact with  Ethereum via a special (ACeD) smart contract; the contract is owned by a group of permissioned oracle nodes, who individually cannot update the contract, except that a majority of the oracle nodes together  agree and perform a single action on the smart contract. The contract only changes its state after accepting a multisignature  of majority votes  \cite{breidenbach2017depth}. We implement both the oracle nodes and side blockchain nodes in {\sf RUST}; the architecture is depicted in Figure~\ref{fig:block-diagram}. There are four important modules:
\begin{figure}[hbt]
   % \vspace{-0.3in}
    \centering
    \includegraphics[width=0.9\textwidth]{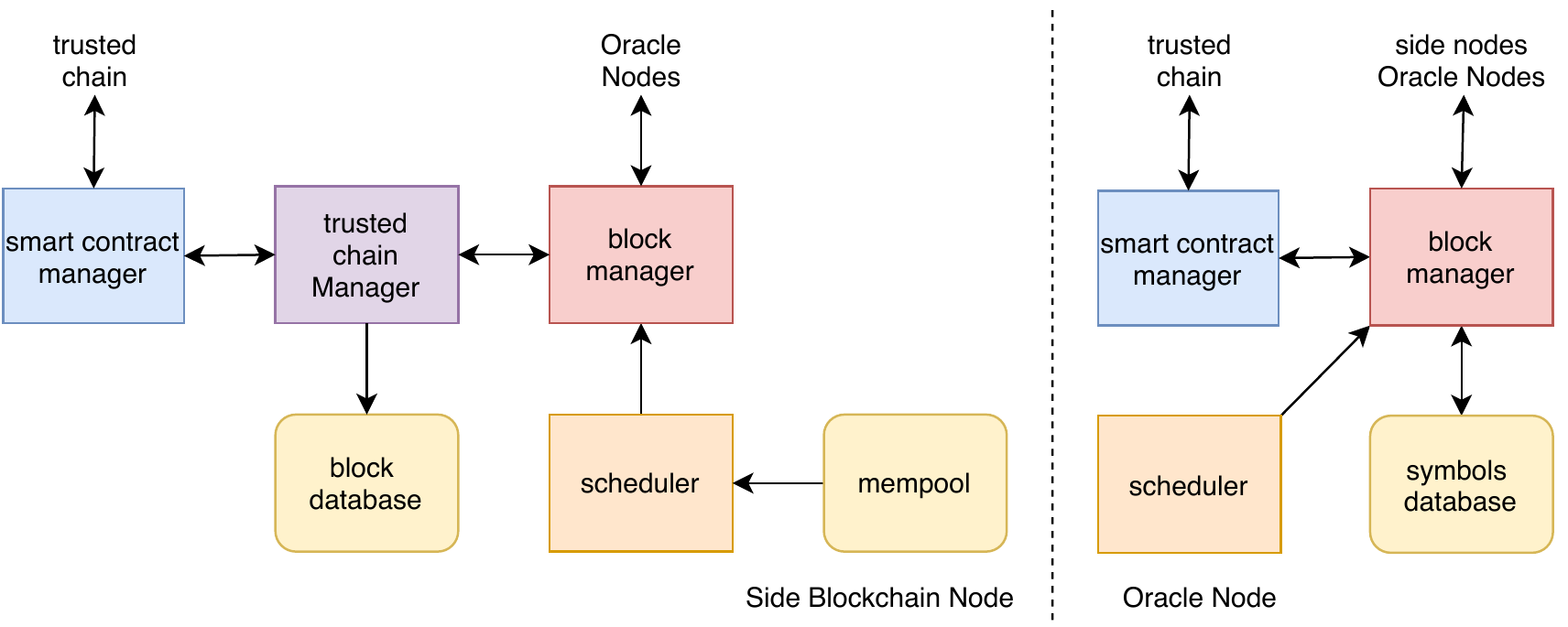}
    \caption{The left figure depicts a block diagram of a side node; the right figure depicts an oracle node. The sharp rectangle represents an active thread; round rectangles are system states.}
    \label{fig:block-diagram}
    \vspace{-0.3in}
\end{figure}

(1)  {\bf Smart contract manager}, when used by a side blockchain node, is responsible for fetching states from the trusted blockchain; when used by an oracle node, it writes state to the trusted blockchain after the block manager has collected sufficient signatures. 
(2)  {\bf Block manager}, when used by a side blockchain node,  creates a block (triggered by a scheduler), and distributes the coded symbols to oracle nodes; when used by an oracle node and has received messages from side blockchain nodes, the module stores the received symbols and provides them on demand; When used by an oracle committee node, the module collects signatures and send an Ethereum transaction once there are sufficient signatures. 
(3) {\bf Scheduler}, when used by a side blockchain node, decides when a node can transmit its block (and coded symbols) to oracle nodes; when used by an oracle node, the module decides who the correct block proposer is.   
(4) {\bf Trusted Chain manager}, when used by a side node,  periodically pulls state changes from the trusted blockchain; if a new block is updated, it collects block symbols from all oracle nodes; the module also maintains a blockchain data structure.

A detailed discussion of the design and implementation optimizations of these four blocks is provided in Appendix~\S\ref{sec:implementation-appendix}, which also discusses the data structures that maintain the state of ACeD. 

\noindent {\bf Implementation Optimizations}. 
The key challenge to a high performance implementation of  ACeD is the large number of interactions among the various system blocks. Since the trusted blockchain is immutable, we optimize the oracle and side blockchain  client implementations by utilizing parallelism and pipelining as much as possible. % which uses parallelism to minimize the latency, and pipeling to decouple the throughput and latency. 
(1) {\bf{LDPC parallelization}}\ref{sec:erasure-code}: we bring state of the art techniques from wireless communication design and encoding procedures (Appendix A of \cite{richardson2008modern}) to our implementation of encoding a block by distributing the workload to many worker threads; we find that the encoding time is almost quartered using 4 threads. A  detailed description is referred to Appendix~\S\ref{apdx:Ldpc}.
(2) {\bf Block Pipelining}. We pipeline signature aggregation, block dispersal and retrieval so that multiple blocks can be handled at the same time. We also pipelined the trusted chain manager so that a side node can simultaneously request multiple blocks from oracle nodes. The local state variable, (block id, hash digest), is validated against the smart contract state whenever a new block is reconstructed locally. %In section 4.4, we described a summary about how one block is processed in the ACeD protocol. In the simplest implementation, a new block is generated only after its previous block gets completed, then the latency and throughput are interlocked because once transaction is fixed the only way to increase the throughput once number of transaction is to reduce the block confirmation latency. The natural solution is to allow multiple blocks being processed simultaneously in ACeD. We added a key-value map in oracle clients, so each node can perform signature aggregation, block dispersal and retrieval for multiple blocks at the same time. We also pipelined the trusted chain manager so that a side node can simultaneously request multiple blocks from oracle nodes. The local state variable, (block id, hash digest), is validated against the smart contract State whenever a new block is reconstructed locally.

\section{Evaluation}
\label{sec:evaluation}
We aim to answer the following questions through our evaluation. (a) What is the highest throughput (block confirmation rate) that ACeD can offer to the side blockchains in a practical integration with Ethereum?  (b) What latency overhead does ACeD pose to block confirmation?  (c) How much is the gas cost reduced, given the computation (smart contract execution) is not conducted on the trusted blockchain? (d) What are the bottleneck elements to performance? 
%The evaluation tries to understand how block generation rate and block size affect the throughput and gas cost. We also break down CPU allocation for each node to each task to understand the bottleneck of the system.

{\bf{Testbed}}. We deploy our implementation of ACeD  on Vultr Cloud hardware of 6 core CPU, 16GB memory, 320GB SSD and a 40Gpbs network interface (for both oracle and side blockchain nodes). % We setup an experiment including 15 instances, and allocate 10 instances to be the Oracle nodes and the rest 5 are side chain nodes running on the same applications.
 To simplify the oracle layer, we consider all of oracle nodes are committee members, which is reflected in the experiment topology: all the oracle nodes are connected as a complete graph, and each side blockchain node has TCP connections to each of the oracle nodes.
%In addition, all nodes in the experiment follow the protocol, and therefore we used an simple dispersal protocol that distributes symbols equally to all nodes with $\lambda=1$.
 We deploy a smart contract written in {\sf Solidity} using the  {\sf Truffle} development tool on a popular Ethereum testnet: {\sf Kovan}. %All evaluation results are obtained from the Kovan testnet for its 4 seconds block interarrival time, which is much faster than Rospten testnet. In order for 
 Oracle nodes write their smart contracts using  Ethereum accounts  created with {\sf MetaMask},  each  loaded with sufficient Keth to pay the gas fee of each Ethereum transaction. %We used a monitoring  platform by periodically sends API requesting for collecting snapshot from each live instance. By default 
 We use an LDPC code of $\frac{b}{c} = 128$ input symbols with a coding rate of $r=0.25$ to construct CIT with $q=8, d=8,\alpha=0.125,\eta=0.875, t=16$. We fix the transaction length to be 316 bytes, which is sufficient for many Ethereum transactions.  
 
%\noindent {\bf  Results}. 

\noindent {\bf  Experiment Settings}. We consider four separate experiments with varying settings of four key parameters: the number of oracle nodes, the number of side blockchain nodes, block generation rate, and  block size. The experiment results are  in Figure~\ref{fig:measurement}.
\begin{table}[htbp]
    \centering
    \scalebox{0.92}{
    \begin{tabular}{|c|c|c|c|c|}
         \hline
     &   \# side blockchain nodes & \# oracle nodes & Block size(MB) & Block generation rate(sec/blk) \\
     \hline
       A  & 5  & 5,10,15,25 & 4 & 5  \\
       \hline
       B & 5,10,20 & 10 & 4 & 5 \\
       \hline
        C & 5 & 10 & 4,8,16 & 5 \\ 
        \hline
        D& 3,5,8,10 & 10 & 4 & 8.33,5,3.125,2.5\\
        \hline
    \end{tabular}}
    \caption{Four different experiments  varying the parameters of ACeD. }
    \label{tab:my_label}
\end{table}

\noindent (1) {\bf Throughput}. We measure the  rate of state update in the trusted blockchain as viewed from each oracle blockchain nodes; the throughput is then the  rate averaged across time and across oracle blockchain nodes.  The throughput performance changes with four parameters:  %Our experiment   results are depicted in  Figure~\ref{fig:measurement}. 
In experiments A and B, the  throughput is not noticeably impacted as the number of oracles or side blockchain nodes increases. In experiment C, the block size has roughly a linear effect on throughput, which continues  until coding the block is too onerous (either it costs too much memory or takes too long).  In experiment D, we fix the product of the block generating rate and the number of side blockchain node to constant, while increasing the block generation rate, the throughput increases linearly; the performance will hit a bottleneck when the physical hardware cannot encode a block in a time smaller than the round time.

\noindent (2) The {\bf latency} of ACeD is composed of three major categories:  block encoding time, oracle vote collection time and  time to update the trusted blockchain. %Because the blockchain update time is random in nature, the Figure~\ref{fig:measurement} Latency plots the sum of encoding and oracle vote collection time under various experiments. 
We find latency stays relatively constant in all experiments (the exception is experiment C where   block encoding is onerous). 

\noindent (3) {\bf Gas saving}. ACeD transactions cost significantly less gas than a regular Ethereum transaction, independent of the computational complexity of the transaction itself. The cost of an ACeD transaction voted by 10 oracles nodes on a 4MB block costs on average 570K gas.  Based on a historical analysis of 977 Crytokitties transactions, we estimate a 220 byte transaction costing  roughly 180,000 gas: ACeD gas savings are thus over a factor 6000. We emphasize that the saving in gas is independent of block generation rate or the block size, since the smart contract only needs the block header and a list of the oracle signatures. % to verify with. 

\begin{figure}[hbt]
    \centering
    \includegraphics[width=0.75\linewidth]{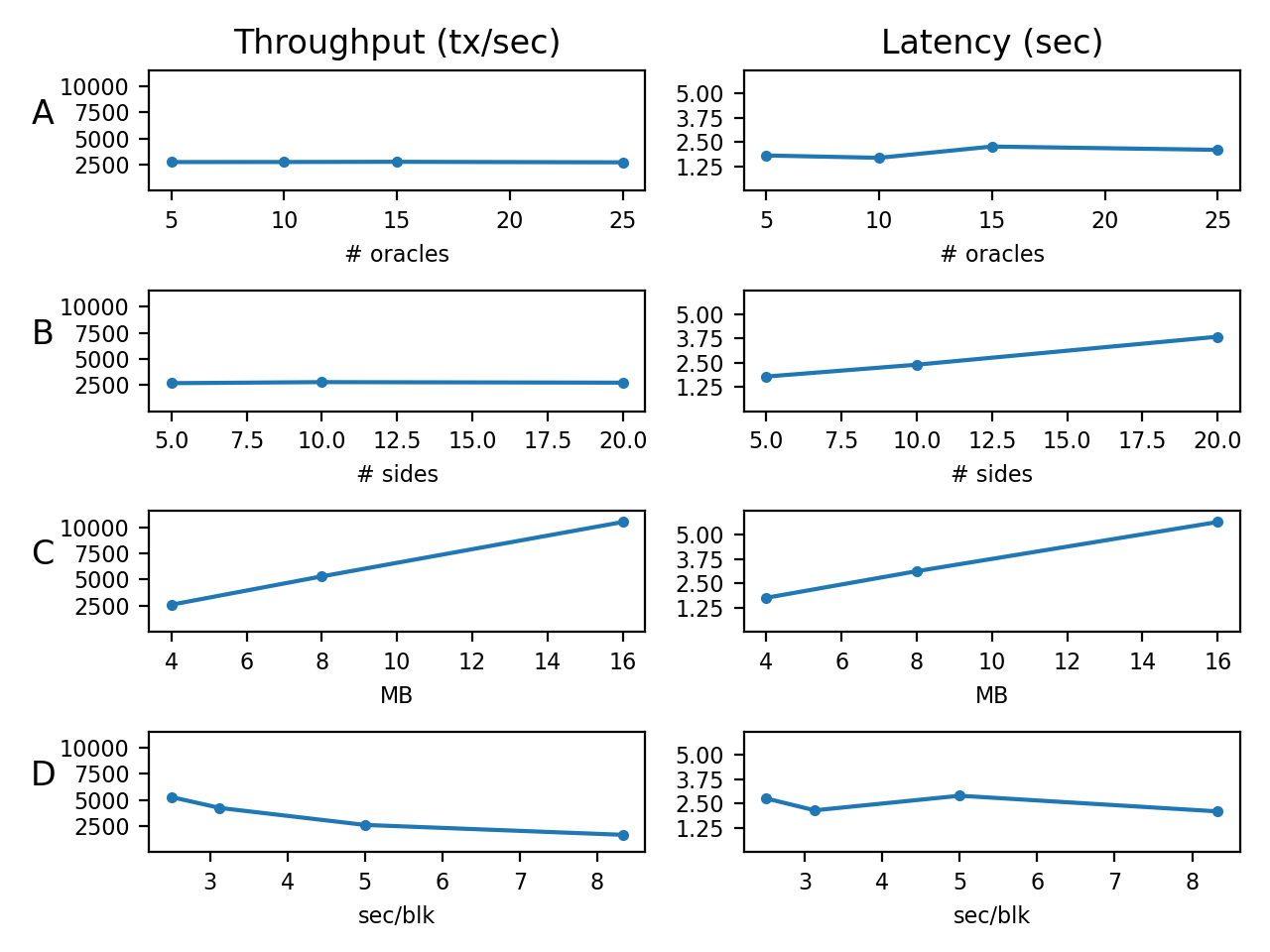}
    \vspace{-0.2in}

    \caption{Throughput (left) and Latency (right) for the 4 experiments A, B, C, D. }
    \label{fig:measurement}
\end{figure}

%\begin{figure}[hbt]
%    \centering
%    \includegraphics[width=0.5\linewidth]{figures/pie-latency-15.png}
%    \caption{Latency decomposition to three components: time to encode a block, time for Oracle node to aggregate sufficient signature, time required by the trusted chain to accept a transaction}
%    \label{fig:latency}
%\end{figure}

%{\bf{Single Block latency dependency on block generation rate}}. The block generation rate does not affect the latency for a block to get updated on the trusted chain if block is generated at high rate. [Need to run experiment to confirm the logic, we can plot latency vs. block generation rate].

%{\bf{Throughput dependency on block generation rate}}. Since ACeD can parallelize ACeD dispersal and retrieval, we increase block generation rate until hit some bottleneck.  [The experiment is the same as above]

%{\bf{Throughput and latency dependence on block size}}. We can also vary the block size to see 

%\subsection{Resource Utilization}
%[TBA] Use flame graph to decompose CPU usage.
\section{Conclusion and Discussion}\label{sec:conclusion}
Interoperability across blockchains is a major research area. Our work can be viewed as a mechanism to transfer the  trust of an established blockchain (eg: Ethereum) to  many small (side) blockchains. Our proposal, ACeD, achieves this by building an oracle that guarantees data availability. The oracle layer enabled by ACeD is run by a group of permissioned nodes which are envisioned to provide an {\em interoperability service} across blockchains; we have provided a detailed specification of ACeD, designed incentives to support our requirement of an honest majority among the oracle nodes, and designed and evaluated a high performance implementation. 

ACeD solves the data availability  problem with near optimal performance on all metrics. However, when instantiating the system with specific parameters, we observe that some metrics have to be computed with constant factors which give rise to acceptable but non-negligible performance loss. The critical bottleneck for ACeD is the lower undecodable ratio (the fraction of symbols in the minimum stopping set) compared to 1D-RS coding;  this undermines dispersal efficiency and  increases the communication cost under the existence of strong adversary. Therefore, finding a {\em deterministic} LDPC code with a  higher undecodable ratio will immediately benefit the practical performance of ACeD; the construction of such LDPC codes  is an exciting direction of future work  and is of independent interest.     
\section{Acknowledgements}
This research is supported in part by a gift from IOHK Inc., an Army Research Office grant W911NF1810332 and by the National Science Foundation under grant CCF 1705007.

% \bibliography{splncs}

\begin{thebibliography}{10}
\providecommand{\url}[1]{\texttt{#1}}
\providecommand{\urlprefix}{URL }
\providecommand{\doi}[1]{https://doi.org/#1}

\bibitem{aced}
Aced library, \url{https://github.com/simplespy/ACeD.git}

\bibitem{al2018fraud}
Al-Bassam, M., Sonnino, A., Buterin, V.: Fraud and data availability proofs:
  Maximising light client security and scaling blockchains with dishonest
  majorities. arXiv preprint arXiv:1809.09044  (2018)

\bibitem{albassam2019fraud}
Al-Bassam, M., Sonnino, A., Buterin, V.: Fraud and data availability proofs:
  Maximising light client security and scaling blockchains with dishonest
  majorities (2019)

\bibitem{alon2004probabilistic}
Alon, N., Spencer, J.H.: The probabilistic method. John Wiley \& Sons (2004)

\bibitem{badertscher2018ouroboros}
Badertscher, C., Ga{\v{z}}i, P., Kiayias, A., Russell, A., Zikas, V.: Ouroboros
  genesis: Composable proof-of-stake blockchains with dynamic availability. In:
  Proceedings of the 2018 ACM SIGSAC Conference on Computer and Communications
  Security. pp. 913--930 (2018)

\bibitem{bagaria2019prism}
Bagaria, V., Kannan, S., Tse, D., Fanti, G., Viswanath, P.: Prism:
  Deconstructing the blockchain to approach physical limits. In: Proceedings of
  the 2019 ACM SIGSAC Conference on Computer and Communications Security. pp.
  585--602 (2019)

\bibitem{breidenbach2017depth}
Breidenbach, L., Daian, P., Juels, A., Sirer, E.G.: An in-depth look at the
  parity multisig bug. 2017. URL: http://hackingdistributed.
  com/2017/07/22/deepdive-parity-bug  (2017)

\bibitem{buterin2017casper}
Buterin, V., Griffith, V.: Casper the friendly finality gadget. arXiv preprint
  arXiv:1710.09437  (2017)

\bibitem{cachin2005asynchronous}
Cachin, C., Tessaro, S.: Asynchronous verifiable information dispersal. In:
  24th IEEE Symposium on Reliable Distributed Systems (SRDS'05). pp. 191--201.
  IEEE (2005)

\bibitem{csiszar1998method}
Csisz{\'a}r, I.: The method of types [information theory]. IEEE Transactions on
  Information Theory  \textbf{44}(6),  2505--2523 (1998)

\bibitem{daian2019snow}
Daian, P., Pass, R., Shi, E.: Snow white: Robustly reconfigurable consensus and
  applications to provably secure proof of stake. In: International Conference
  on Financial Cryptography and Data Security. pp. 23--41. Springer (2019)

\bibitem{david2018ouroboros}
David, B., Ga{\v{z}}i, P., Kiayias, A., Russell, A.: Ouroboros praos: An
  adaptively-secure, semi-synchronous proof-of-stake blockchain. In: Annual
  International Conference on the Theory and Applications of Cryptographic
  Techniques. pp. 66--98. Springer (2018)

\bibitem{decker2015fast}
Decker, C., Wattenhofer, R.: A fast and scalable payment network with bitcoin
  duplex micropayment channels. In: Symposium on Self-Stabilizing Systems. pp.
  3--18. Springer (2015)

\bibitem{easley2010networks}
Easley, D., Kleinberg, J., et~al.: Networks, crowds, and markets, vol.~8.
  Cambridge university press Cambridge (2010)

\bibitem{zkrollup}
EthHub.: Zk-rollups.,
  \url{https://docs.ethhub.io/ethereum-roadmap/layer-2-scaling/zk-rollups/}

\bibitem{fitzi2018parallel}
Fitzi, M., Gazi, P., Kiayias, A., Russell, A.: Parallel chains: Improving
  throughput and latency of blockchain protocols via parallel composition. IACR
  Cryptol. ePrint Arch.  \textbf{2018}, ~1119 (2018)

\bibitem{foti2020blockchain}
Foti, A., Marino, D.: Blockchain and charities: A systemic opportunity to
  create social value. In: Economic and Policy Implications of Artificial
  Intelligence, pp. 145--148. Springer (2020)

\bibitem{gilad2017algorand}
Gilad, Y., Hemo, R., Micali, S., Vlachos, G., Zeldovich, N.: Algorand: Scaling
  byzantine agreements for cryptocurrencies. In: Proceedings of the 26th
  Symposium on Operating Systems Principles. pp. 51--68 (2017)

\bibitem{jirgensons2018blockchain}
Jirgensons, M., Kapenieks, J.: Blockchain and the future of digital learning
  credential assessment and management. Journal of Teacher Education for
  Sustainability  \textbf{20}(1),  145--156 (2018)

\bibitem{kalodner2018arbitrum}
Kalodner, H., Goldfeder, S., Chen, X., Weinberg, S.M., Felten, E.W.: Arbitrum:
  Scalable, private smart contracts. In: 27th $\{$USENIX$\}$ Security Symposium
  ($\{$USENIX$\}$ Security 18). pp. 1353--1370 (2018)

\bibitem{kamath1995tail}
Kamath, A., Motwani, R., Palem, K., Spirakis, P.: Tail bounds for occupancy and
  the satisfiability threshold conjecture. Random Structures \& Algorithms
  \textbf{7}(1),  59--80 (1995)

\bibitem{kiayias2017ouroboros}
Kiayias, A., Russell, A., David, B., Oliynykov, R.: Ouroboros: A provably
  secure proof-of-stake blockchain protocol. In: Annual International
  Cryptology Conference. pp. 357--388. Springer (2017)

\bibitem{kokoris2018omniledger}
Kokoris-Kogias, E., Jovanovic, P., Gasser, L., Gailly, N., Syta, E., Ford, B.:
  Omniledger: A secure, scale-out, decentralized ledger via sharding. In: 2018
  IEEE Symposium on Security and Privacy (SP). pp. 583--598. IEEE (2018)

\bibitem{luu2016secure}
Luu, L., Narayanan, V., Zheng, C., Baweja, K., Gilbert, S., Saxena, P.: A
  secure sharding protocol for open blockchains. In: Proceedings of the 2016
  ACM SIGSAC Conference on Computer and Communications Security. pp. 17--30
  (2016)

\bibitem{miller2017sprites}
Miller, A., Bentov, I., Kumaresan, R., McCorry, P.: Sprites: Payment channels
  that go faster than lightning. CoRR abs/1702.05812  \textbf{306} (2017)

\bibitem{pass2017hybrid}
Pass, R., Shi, E.: Hybrid consensus: Efficient consensus in the permissionless
  model. In: 31st International Symposium on Distributed Computing (DISC 2017).
  Schloss Dagstuhl-Leibniz-Zentrum fuer Informatik (2017)

\bibitem{poon2017plasma}
Poon, J., Buterin, V.: Plasma: Scalable autonomous smart contracts. White paper
  pp. 1--47 (2017)

\bibitem{rana2020free2shard}
Rana, R., Kannan, S., Tse, D., Viswanath, P.: Free2shard:
  Adaptive-adversary-resistant sharding via dynamic self allocation. arXiv
  preprint arXiv:2005.09610  (2020)

\bibitem{richardson2008modern}
Richardson, T., Urbanke, R.: Modern coding theory. Cambridge university press
  (2008)

\bibitem{teutsch2019scalable}
Teutsch, J., Reitwie{\ss}ner, C.: A scalable verification solution for
  blockchains. arXiv preprint arXiv:1908.04756  (2019)

\bibitem{yang2019prism}
Yang, L., Bagaria, V., Wang, G., Alizadeh, M., Tse, D., Fanti, G., Viswanath,
  P.: Prism: Scaling bitcoin by 10,000 x. arXiv preprint arXiv:1909.11261
  (2019)

\bibitem{yu2020ohie}
Yu, H., Nikoli{\'c}, I., Hou, R., Saxena, P.: Ohie: Blockchain scaling made
  simple. In: 2020 IEEE Symposium on Security and Privacy (SP). pp. 90--105.
  IEEE (2020)

\bibitem{yu2020coded}
Yu, M., Sahraei, S., Li, S., Avestimehr, S., Kannan, S., Viswanath, P.: Coded
  merkle tree: Solving data availability attacks in blockchains. In:
  International Conference on Financial Cryptography and Data Security. pp.
  114--134. Springer (2020)

\end{thebibliography}

\appendix
\section{Proof of Lemmas and Theorems}

\subsection{Proof of Lemma~\ref{thm:dispersal}}
\label{sec:proofdispersal}

We show the existence of (several) chunk dispersal algorithms using the probabilistic method \cite{alon2004probabilistic}. 
Consider a randomized chunk dispersal design, where each element of each $A_i$ is chosen i.i.d. uniformly randomly from the set of $M$ possible chunks. This gives raise to a randomized code $\mathcal{C}$.  
We  prove the following statement:

\begin{equation}
    \mathbb{P}\{\mathcal{C} \text{ is NOT a } (M,N,\gamma,\eta,\lambda) \text{ code}\} \leq e^{-n}
\end{equation}

Let $k = M/(N \lambda)$ be the number of chunks at a given node. 

\begin{eqnarray}
    \mathbb{P}\{\mathcal{C} \text{ is not a valid code} \} & = & P( \exists S \text{ with } |S| = \gamma N:  |\bigcup_{i \in S} A_i| \leq \eta M) \\
    &\leq &
    \sum_{S \subseteq [M]: |S| = \gamma N}P( |\bigcup_{i \in S} A_i| \leq \eta M) \\
    &\leq &
    e^{N H_e(\gamma)}P( |\bigcup_{i \in S} A_i| \leq \eta M)
\end{eqnarray}
where we have used that $\binom{N}{\gamma N} \leq 2^{NH(\gamma)} = e^{NH_e(\gamma)} $, with $H(.)$ being the binary entropy function measured in bits (logarithm to the base $2$) and $H_e$ being the binary entropy measured in natural logarithm (nats), a standard inequality from the method of types \cite{csiszar1998method}.

The key question now is for a fixed $S$, we need to obtain an upper bound on $P( |\bigcup_{i \in S} A_i| \leq \eta M)$. 
We observe that under the random selection, we are choosing $ \gamma / \lambda M$ chunks (sampled randomly with replacement). The mathematical question now becomes, if we choose  $\rho M$ chunks randomly (with replacement)  from $M$ chunks, what is the probability that we get at least $\eta M$ distinct chunks. 
\begin{lemma} We sample $\rho M$ chunks from a set of $M$ chunks randomly with replacement. Let $Y$ be the number of distinct chunks. 
Let $x = (1-\eta) exp(\rho)$.

\begin{equation}
 \mathbb{P}(Y < \eta M) \leq  e^{- (1-\eta) \frac{(x-1)^2}{x(x+1)}}
\end{equation}

\end{lemma}
\begin{proof}
Let $Z = M-Y$ be the number of chunks not sampled. Then $\mu := \mathbb{E}[Z] = (1-\frac{1}{M})^{\rho M} \approx M e^{-\rho}$ (when $M$ is large). We can write $Z = \sum\limits_{i=1}^M Z_i$, where $Z_i$ is the binary indicator random variable indicating that chunk $i$ is not sampled. We note that a direct application of tail bounds for i.i.d. random variables does not work here. However, there are results analyzing this problem in the context of balls-in-bins and they show that the following tail bound does indeed hold (see Theorem~$2$ and Corollory~$1$ in \cite{kamath1995tail}). We can write the tail bound on $Z$ as follows: 
\begin{equation}
     P(Z > \ell\mu) \leq e^{-\frac{(\ell-1)^2}{1+\ell} \mu} \quad \forall \ell >1
\end{equation}
\end{proof}

Let $\bar{\eta} = 1-\eta$ and $\ell = e^{\rho}\bar{\eta}$, we get 
\begin{equation}
     P(Z > \bar{\eta} M) \leq e^{-\frac{(e^{\rho}\bar{\eta}-1)^2}{e^{\rho}(e^{\rho} \bar{\eta} +1)}M} \quad \text{if } e^{\rho} \bar{\eta} >1
\end{equation}

Define the function, $f(\eta,\rho) = \frac{(e^{\rho}\bar{\eta}-1)^2}{e^{\rho}(e^{\rho} \bar{\eta} +1)}$ (we note that $f(\cdot)$ is a positive function), we have $P(Y < \eta M) \leq e^{-f(\eta,\rho) \cdot M } $ if $\rho > \log (\frac{1}{1-\eta})$.

Continuing from before, we get 
\begin{eqnarray}
    \mathbb{P}\{\mathcal{C} \text{ is not a valid code} \} 
    &\leq &
    e^{N \cdot H_e(\gamma)-M \cdot f(\eta,\rho)} \text{ if } \frac{\gamma}{\lambda} > \log \left(\frac{1}{1-\eta} \right)
\end{eqnarray}
If we fix a $N$ and take $M$ large enough, we can make the right hand side large enough to make the probability that $\mathcal{C}$ is not a valid code vanish to as small as we want. This result has two interpretations: (1) as long as the probability of being a valid code is strictly greater than zero, then it proves the existence of a {\em deterministic} chunk distribution algorithm, (2) if we can make the probability of being a invalid code arbitrarily close to zero, then we can use the randomized algorithm directly with a vanishing probability of error. Fixing a $N$ and taking $M$ large can make the RHS arbitrarily small and thus the stronger second interpretation can be used.  This concludes the proof of the theorem.

\begin{comment}\subsection{Proof of Lemma \ref{lm:committee}}
\label{proof-lm-committee}
\begin{proof}
Suppose there are $N$ oracle nodes where $N = 2^n$, the maximum possible ratio of byzantine nodes is denoted as $\beta_{max}$. Let $H$ be the subset of honest nodes, $\frac{|H|}{N}\geq (1-\beta_{max})$ . Let $C$ be a uniformly sampled subset with cardinality $\lceil cn\rceil$, where $c$ is a constant and $c > ln(2) / (1-\beta_{max})$, then
    \begin{align*}
    &P[H \cap C \ne \emptyset] =1-\frac{\binom{N-|H|}{\outcomm}}{\binom{N}{\outcomm}}\\
    &=1-( \frac{N-\outcomm}{N}\cdot\frac{N-1-\outcomm}{N-1}\cdots\frac{N-|H|+1-\outcomm}{N-|H|+1})\\
    &\ge 1-(1-\frac{\outcomm}{2^n})^{|H|}\ge 1-\exp(-(1-\beta_{max})\outcomm) \\
    &\ge 1-\exp(-(1-\beta_{max})cn)\\
    &\ge 1-\frac{1}{2^n}
\end{align*}
Notice that adversary is static, after initial randomization, they are assumed to be uniformly distributed among oracle nodes.
\end{proof}

\end{comment}
\subsection{Proof of Lemma \ref{lm:recon}}
\label{proof-lm-recon}

\begin{proof}
By definition, $f_j(i)$ can be decomposed to two functions, $p_j(i)$ and $e_j(i)$. According to equation \ref{eq:pe}, $p_j(i)$ is computed by modulo, so there are at most $c = \frac{m_{\ell}}{rm_{j}} = \frac{m_{j} \cdot (qr)^{\ell-j}}{rm_{j}} = \frac{ (qr)^{\ell-j}}{r}$ base symbols mapping to one symbol in layer $j$ . In the worst case, $\gamma m_\ell$ distinct base layer symbols map to $\frac{\gamma m_\ell}{c} = \gamma r m_j$ distinct symbols in layer j in the range $[0, rm_j]$. Similarly, there are $d = \frac{M_{\ell}}{(1-r)M_{j}} = \frac{M_{j} \cdot (qr)^{\ell-j}}{(1-r)m_{j}}=\frac{(qr)^{\ell-j}}{1-r}$ base layer symbols mapping to one symbol at layer $j$, and therefore in the worst case there are $\frac{\gamma m_\ell}{d} = \gamma(1-r)m_j$ distinct symbols in the range $[rm_j, m_j)$. Since two range do not overlap, in total there are $\gamma m_j$ distinct symbols for each layer. Hence it completes the proof.
\end{proof}

\subsection{Proof of Lemma \ref{lm:sibl}}
\label{proof-lm-sibl}

\begin{proof}
First we show for every $i\in[m_\ell]$ and $2\le j < \ell $, $p_j(i)$ and $e_j(i)$ are sibling. According to the definition,
\begin{align}
	p_{j}(i) &= i \bmod rm_{j-1} = i \bmod qr^2m_{j-2}\\
	p_{j-1}(i) &= i \bmod rm_{j-2}  \quad \text{(by definition)}\\
		       &= p_j(i) \bmod rm_{j-2} \quad \text{(discuss below)}\label{eq:3}
\end{align}
To prove equation \ref{eq:3}, suppose $i = p_{j}(i) + xqr^2m_{j-2}$ and $ i = p_{j-1}(i) + yrm_{j-2}$ for some integer $x, y$. By rearranging terms,  $p_{j}(i) - p_{j-1}(i) = rm_{j-2}(y-xqr)$. If $p_{j}(i) - p_{j-1}(i) \neq 0$, we mod $rm_{j-2}$ on both sides, and therefore $p_{j-1}(i) = p_{j}(i) \bmod rm_{j-2}$. If $p_{j}(i) - p_{j-1}(i) = 0$, then obviously $p_{j-1}(i) = p_{j}(i) \bmod rm_{j-2}$. Therefore in any case, equation \ref{eq:3} holds, which means $p_{j-1}(i)$ is the parent of $p_{j}(i)$.

Then let's see another function $e_j(i)$, by definition we have
\begin{align*}	
	e_{j}(i) &= rm_{j-1} + (i \bmod (1-r) m_{j-1})\\
		   &= rm_{j-2} qr + (i \bmod (1-r) qr m_{j-2})
\end{align*}
To calculate its parent according to aggregation rules, we get similar result that $e_{j}(i) - p_{j-1}(i) = rm_{j-2}(y-xq(1-r))$ thus
\begin{align*}
	e_{j}(i)  &= 	i \bmod (1-r) qr M_{j-2}\\
 &= 	p_{j-1}(i) \bmod rm_{j-2} \\
\end{align*}
$p_{j-1}(i)$ is also the parent of $e_{j}(i)$. Hence it completes the proof.

\end{proof}

\section{Efficient Encoding for LDPC Codes}
\label{apdx:Ldpc}
A block of transactions is coded to symbols before distributing them to the oracle nodes. We have implemented a CIT library in {\sf RUST}; compared to CMT (coded Merkle tree) library, the major distinction is in the data interleaving algorithm. We find  that the encoding time of CMT is very time consuming: on a Mac 2.6 GHz 6-Core Intel Core i7 computer, the library required about 34 seconds to encode a 4 MB block on a basis of 128 symbols with 0.25 coding ratio. The issue is due to a huge amount of XOR operations required for encoding the base layer. In our implementation, we have developed two strategies to overcome the problem by adding parallelism to the encoding process and by devising a code transformation that reduces coding complexity. Parallelization involves two components: a light weight read-only hash map and a multiple-thread environment for partitioning the workload. An inverse map is created whose key is the index of input symbols, and value is a list of parity symbols to be XOR with according to the encoding matrix. We use a single-producer multiple-workers paradigm where the producer thread uses the inverse map to properly distribute the input symbol. Worker threads partition the parity set such that each thread is responsible for an equal number of parity symbols. The input symbol is communicated through a message queue, and once a worker receives the input symbol, it XORes the symbol locally. The producer  distributes all input symbols  and  sends a final control message to all workers to  collect the result. The optimization brings out a significant speedup,  reducing the encoding time of a 4MB block from 31 second to about 10 sec using 4 threads. Adding as many as 10 threads further reduces the encoding time to 8.5 seconds, and eventually hitting  physical (CPU) limits. In the second strategy, we attempted to change Coding matrix to reduce the number of XOR computation at the root. The conventional encoding complexity of LDPC codes is $O(n^2)$ where $n$ is the length of coded symbols. We could transform the original encoding matrix to a upper triangular form that has a computation complexity of $O(n)$. But after a careful analysis, we discover an extra constant factor overkills the reduced coding complexity, so in the end we rely solely on parallelism to reduce coding latency.

\section{Incentive Analysis}
\label{sec:incentiveanalysis}
\subsection{Motivation and Countermeasures} For oracle nodes, downloading and checking data requires nontrivial computation effort, selfish oracle nodes have an incentive to follow others' decisions and skip verification to compete for more reward by voting for more blocks. Such a phenomenon, known as {\em  information cascade} in the literature,  sets forth that people tend to make decisions according to the inferences based on the actions of earlier people. As a consequence, an oracle node may submit a vote for a block but can not provide any data when a retrieval request is received. To solve this problem, we introduce an audit scheme. When a block is added to the side blockchain, with probability $p_a$, a voted node will be selected to submit the received data chunks. Our dispersal protocol guarantees that each oracle node ought to store a specific subset of data chunks. For those committed blocks which are not successfully appended in the side blockchain, all voted nodes should submit data or lose their stakes. 

Even if there are several nodes in the aggregation committee, only the first one who submits the valid signature can claim the reward, thus rational committee members would choose to submit votes containing the first bath of signatures. And since only those oracle nodes whose signatures are used to aggregate the final signature would receive a reward, nodes with the largest network latency may never get any income. Furthermore, adversaries can violate an arbitrary subset of honest nodes' interests by excluding them intentionally. We can not influence the adversarial behavior, but to encourage selfish nodes to include more signatures, the committee reward will be proportional to the number of signatures in final aggregation.

\subsection{Model}
We consider a network to be a set of $N$ configured oracle nodes (oracle nodes) $O = \{o_1 , o_2, \cdots, o_N \}$.  Suppose $\beta N$ nodes in the network are \textit{Byzantine}, who behave arbitrarily. And that every non-Byzantine nodes are \textit{Rational}, they follow the protocol if it suits them, and deviate otherwise. In the registration phase, each oracle node needs to deposit an amount of stake $stk_o$ for registering a pair of valid BLS key set.

There is a set of blocks proposed by side nodes. To propose a block, the proposer needs to deposit an amount of stake $stk_b$ and will receive block reward $B$ if successful. A fraction of block reward $\eta B$ will be used for rewarding oracle nodes. There is a small committee $C \subset O$ in the oracle layer. If a block is committed, the committee member who submits the commitment will get an extra reward, which comes from a constant submission fee $r_m$ paid by other oracle nodes.
 
 Basically there are three general actions a node can take, cooperate, offline and defect, denoted as $D,O,C$ respectively. We define $\allc$, $\allc$ and $\alld$ representing the strategy that all nodes choose to be cooperative, offline and defective. We design the utility functions for the block proposer, oracle nodes and committee members when a block is successfully committed in Table.\ref{tab:utility}.

\begin{table}[]
    \centering
    \begin{tabular}{|c|c|c|c|}
    \hline
       strategy & block proposer &  oracle nodes & committee member\\ \hline
        $D$ & $-stk_b$ &$-p_a   stk_v + (1-p_a) \frac{\eta B}{k} - r_m$ &$-stk_m$ \\ \hline
        $O$ & $0$ & $0 $& $0$\\\hline
        $C $&$(1-\eta)B$ & $\frac{\eta B}{k} - r_m - c_s$ & $kr_m-c_m$ \\ \hline
    \end{tabular}
    \caption{Utility functions for different types of nodes. $k$ is the number of signatures aggregated in committed signature.}
    \label{tab:utility}
\end{table}

\subsection{Incentive Compatibility}
A protocol is incentive compatible if rational actors see no increase in their utility from unilaterally deviating from the protocol. To prove that, there are several things rational participants may care about.

For oracle nodes, there is a reward for having a block it successfully voted for added to the side blockchain. They cannot do much to influence the reward, except creating more voting accounts which means to store more data and deposit more stakes. For the block producer, there is a reward $(1-\eta)B$ when the proposed block is added to the side blockchain. Rational side nodes can't gain more rewards by deviating the protocol thus have no reason to violate the protocol. For the committee member, they claim rewards when a vote is accepted and more participants in aggregation bring more reward to the submission fee. 

We prove that $\allc$ strategy is a  Nash Equilibrium below. 

\begin{theorem}
\label{thm:incentive}
There exists a strategy $\{e_{i}^{*}\}_{i\in[1,n]}$ for all rational oracle nodes in system to reach Bayesian Nash Equilibrium: $\forall k\in [1,n]$, k is rational,
$$\mathbb{E}[U_{k}(e_k^*)|\{e_{i}^{*}\}_{i\in[1,n]}]\ge \mathbb{E}[U_{k}(e_k)|\{e_{i}^{*}\}_{i\ne k}, e_{k} ]$$
and the following statements hold:
\begin{enumerate}
    \item $\allo$ strategy is a Bayesian Nash Equilibrium.
    \item $\allc$ strategy is a Bayesian Nash Equilibrium.
\end{enumerate}
\end{theorem}

\begin{proof}
A protocol is a Byzantine Nash Equilibrium if rational actors see no increase in their utility from unilaterally deviating from the protocol. Firstly, consider $\{e_{i}^{*}\}_{i\in[1,n]} = \allo$, it can be easily derived that $\mathbb{E}[\allo] = 0$. Suppose any voter $k$ wants to change the protocol, no matter it chooses to cooperate or defect, it can not change the consensus thus gets no reward and even pays some cost for computation when choosing to cooperate, thus $\allo$ is a Bayesian Nash Equilibrium.

Suppose a block is eventually committed. The expected utility for $\allc$ is,
\begin{equation}
    \mathbb{E}[\allc] = \frac{\eta B}{k} - r_m - c_s
\end{equation}
The expected utility for $k$ to choose to offline and defect is,
\begin{equation}
    \mathbb{E}[e_{k}=O] = 0
\end{equation}
\begin{equation}
    \mathbb{E}[e_{k}=D] = -p_a stk_v + (1-p_a) \frac{\eta B}{k} - r_m 
\end{equation}
Letting $\mathbb{E}[\allc]>\mathbb{E}[e_{k}=D]$ and $\mathbb{E}[\allc]>\mathbb{E}[e_{k}=O]$ derives that $\allc$ strategy is a Bayesian Nash Equilibrium when $p_a(stk_v+1)-c_s > 0$ and $\frac{\eta B}{k} > r_m + c_s$. Since both $c_s$ and $r_m$ are constant, we can always choose stake value to make it reach the equilibrium.
\end{proof}

\begin{table}[t]
	\center
	\caption{List of symbols used in incentive analysis}
	\begin{tabular}{c|l}
		\hline
		\textbf{Symbol} & \textbf{Definition} \\ \hline 
		$r_m$ & Rewards per signature for a committee member\\
		$B_j$ & Block rewards for block $j$\\ \hline
		$c_s$ & Costs for oracle nodes to verify data chunks \\
		$c_m$ & Costs for committee members to aggregate signatures\\ \hline
	
		$stk_o$ & Unit stake of oracle nodes \\ 
		$stk_m$ &  Unit stake of committee members\\
		$stk_b$ & Stake of block proposer\\\hline
		$\beta$ & Proportion of byzantine nodes \\
		$\theta$ & Threshold of aggregated signature\\
		$\eta$ & Proportion of block rewards as block proposition cost\\
		$p_a$ & Probability to audit when a new block is added\\
		\hline
		\end{tabular}
\end{table}
\section{Design and Implementation of ACeD Modules}
\label{sec:implementation-appendix}

The {\bf Smart Contract manager} is an event loop which waits for the signal from other managers. Its main purpose is to communicate with the smart contract by either sending Ethereum transactions to write new states to the smart contract, or calling the smart contract to read the current state.

The {\bf Trusted Chain manager} is an event loop used by any side nodes for periodically pulling the latest state of the main chain by asking the smart contract manager. Every side node uses it to monitor any state update from the oracle nodes. An update occurs when the smart contract block id is higher than the block id maintained in the blockchain. Each side blockchain node queries the trusted blockchain periodically to be aware of the state change. When a new state is received, the side node queries all oracle nodes using the block id, for assembling the new block. It then checks Merkle proof in the header to check the block integrity. Then the module uses the new hash digest in the smart contract to verify the integrity of the header. The computation is done by hashing the concatenation of the previous hash digest and the hash digest of the new block header $H$; then compare it with the hash digest in the smart contract. The trusted blockchain manager continues until it completes to the latest state and has a consistent view of the blockchain compared to the states on the trusted blockchain. 

The {\bf scheduler} runs an endless loop which measures the current time to determine the correct side node proposer at its moment. It uses the current UNIX EPOCH to compute the time elapsed from a time reference registered at the contract creation, and divides Slot Duration, a system parameter, to get the current slot id. Therefore both oracle and side node can use slot id and the token ring to determine the valid block proposer at the current time. When a node proposes, its block id is set to the current slot id. After receiving block symbols using the dispersal protocol, oracle nodes verify the message sender using the identical scheduling rule. In the special situation while a block is pending in the network and a new block is proposed in the next slot, the pending block might be updated or rejected by the smart contract due to the block id rule constraint by the smart contract. Because a block id is accepted only if it is monotonically increasing, the releasing of the pending block does not hurt the system in either case: If the pending block arrives to smart contract after the new block, the pending block is excluded because it has a lower block id than the new block; if the pending block is included, there might be some conflicting transactions in the two blocks, but since the transactions in the pending block cannot be arbitrarily updated by adversaries after its release, the adversaries cannot adaptively damage the system. This property allows an optimization to improve the throughput (discussed later). 

The {\bf block manager} is an event loop that sends and receives message among peers. When an oracle node receives the header and coded symbols from a side node, it stores them into the symbol database, and submits its signature of the header to its Committee nodes after it receives a sufficient amount of symbols. If an oracle node is chosen as a committee node, the node receives the signatures and aggregates them until the number suggests that there is a sufficient amount of symbols stored in honest nodes. The block manager asks the smart contract manager to send an Ethereum transaction to write the new state to the smart contract.

The  states of ACeD are maintained in three data structures: 

\begin{enumerate}
	\item {\bf block database}, residing on persistent storage, is used by Side nodes to store decoded blocks reconstructed from the oracle nodes
	\item {\bf symbols database}, residing on persistent storage, is used by oracle nodes to store coded symbols from the Side nodes.
	\item {\bf mempool}, transactions queued in memory to create a block. 
\end{enumerate}

The block database is a key-value persistent storage, where the key is block id and the value is a block $B = (H,C)$ which contains a header $H$ and content $C$. The header is a tuple of $(block id, D)$ where D is the digest of the content $C$. Instead of relying on a cryptographic puzzle for accepting blocks, each oracle node uses a deterministic token ring available from the smart contract to decide which block symbols to accept. The token ring is a list of registered side chain nodes eligible for proposing blocks. A block is considered valid only if it is received at the correct time from the correct side node. %A Symbol database is hold by 

\section{Performance table}
\label{sec:performance-table}

 We compared five data dispersal protocols and analyzed their performances using maximal adversarial fraction and communication complexity; we further articulate the storage overhead for the oracle layer and download overhead for each client in two cases depending on whether there are invalid coding attacks in the system. We define those four key metrics of system performance; see Table \ref{tab:metrics}. Let $N$ be the number of nodes and $b$ be the block size.  ``Maximal adversary fraction" measures the fault tolerant capability of the model. ``Storage overhead" measures the ratio of total storage cost and actual information stored. ``Download overhead" measures the ratio of downloaded data size and reconstructed information size. ``Communication complexity"  measures the number of bits communicated

In the uncoded repetition data dispersal protocol, a block proposer client sends the original block to all oracle nodes to ensure the data availability. The system needs more than $1/2$ fraction of the votes for committing a block digest of constant $y$ bits into the trusted blockchain. The total amount of bits communicated equals to $Nb$, where $b$ is block size and $N$ is the total number of oracle nodes. Hence the communication complexity is $Nb$. In the normal case, when a block is coded correctly, the storage overhead is $Nb/b$ whereas the download overhead is $b/b$. When an adversary uploads an invalid block by distributing different blocks to each node, the worst case for storage overhead is identical to the normal case, because every oracle nodes have to store the data block to make sure the data is available when later adversaries oracle nodes decide to make their block available. The download efficiency is $O(1)$ because as long as a client contacts an honest oracle node, it is sufficient to check the data availability. 

In the uncoded data dispersal protocol, the block proposer client divides the original block into $N$ chunks, one for each oracle node. The system cannot tolerate any adversary because any withheld chunk can make the block unavailable. The total communication complexity is $b$ since a block is equally distributed. In both normal and worst case, the storage overhead is $b/b$ because every node only store $b/N$ chunks. Similarly, in both normal and worst cases, the download overhead is $b/b$ because a client downloads at most $b/N$ chunks from $N$ oracle nodes.

In the 1D-RS data dispersal protocol, at least $1/2$ fraction of oracle nodes need to be honest for them to agree to commit a block digest. The communication complexity is $O(b)$ using a chunk dispersal protocol discussed in section~\ref{sec:aced}. In the normal case, each node only stores $O(b/N)$ number of chunks, and therefore the storage overhead is $O(1)$; similarly when a client node wants to download the block, it collects all chunks from the oracles node with download overhead $O(1)$. In the worst case, when an adversarial client disperses chunks with invalid coding,  oracle nodes in the end would produce an invalid coding proof for proving that the block is invalid to all querying clients. In 1D-RS, the fraud proof consists of all chunks as one needs to decode itself to detect inconsistency. Therefore, the overall storage is $O(N (b/N + \log b))$, when $b$ is large $(b/N >> \log b)$, the overall coding proof size is $O(b)$. For computing the storage overhead, we consider the total information stored in the denominator to be the size of the block digest, which is used by a client to start the block retrieval. When a client requests for an invalid coded block, the client needs to download the incorrect-coding proof to convince itself that the block is invalid, hence the download overhead is the same as the storage overhead.

In the 2D-RS data dispersal protocol, at least $1/2$ fraction of nodes needs to be honest. The communication includes the chunks that make up the entire blocks and Merkle trees, and $O(\sqrt{b})$ of Merkle proof, so in total it takes $O(N\sqrt{b}+b+N\log{b}) = O(b)$ bytes. In the normal case, each node only stores $O(b/N)$ bytes, therefore the storage overhead is $O(1)$, similarly the download overhead is $O(1)$. In the worst case, when an adversary disperses invalid coded chunks, each node needs to download $\sqrt{b}\log b$ bytes for producing the invalid code proof, according to Table 2 of \cite{albassam2019fraud}.

In the AVID protocol, oracle nodes run asynchronously with a security threshold $1/3$ \cite{cachin2005asynchronous}. During the block dispersal, every oracle nodes needs to recollect sufficient symbols for reconstructing the original blocks, hence the total communication complexity is $O(Nb)$. In the normal case, after a block is checked valid at the dispersal phase, each node discard chunks except the ones sent from the block producer, hence the storage overhead is $O(1)$; similarly when another client downloads the data, it only needs to download $O(b)$ data, hence its download overhead is $O(1)$. When there is an invalid encoding attack , it can be detected by AVID, and then the nodes would discard the block. Therefore in the worst case, the storage overhead of AVID has the same performance as its normal case. Similarly, since AVID can detect invalid encoding before committing to the trusted blockchain, the worst case download overhead is $O(1)$, the same as its normal case.

In the ACeD dispersal protocol, at least $1/2$ fractions of oracle nodes need to be honest to commit the block digest. The communication complexity is $O(b)$ by using the dispersal protocol in Section~\ref{sec:aced}. In the normal case, each node only stores $O(b/N)$ chunks hence the storage overhead is $O(1)$, and similarly the download overhead is $O(1)$. In the worst case, when an invalid block is dispersed, oracle layer generates an incorrect-coding proof of size $O(\log b)$ as shown in Theorem \ref{theorem-main}. The incorrect-coding proof is then stored in the trusted blockchain. When a client queries for such block, the proof is replied so that the client can be convinced that the block is invalid. For computing the storage and download overhead in the worst case, we consider both the stored information and the size of the reconstructed data to be a single hash which is the pointer to the data.

\section{Erasure code}
\label{sec:erasure-code}

An erasure code encodes a string of input bytes into a string of output bytes which can tolerate missing bytes or check if bytes are modified by using some special decoding algorithms. Suppose we want to encode a block of $b$ bytes, first we divide $b$ bytes into $k$ symbols(a fixed number chunks of bytes), each of $\frac{b}{k}$ bytes, some padding is used if they do not divide perfectly; then we choose an appropriate code which transforms $k$ symbols to $n>k$ symbols for supporting erasure resistance. This ratio $\frac{k}{n}$ is called coding ratio, and usually it is a constant fraction like $\frac{1}{2}$, $\frac{1}{4}$. Many codes offer erasure resistance property including 1D-RS, 2D-RS\cite{albassam2019fraud}, but we choose LDPC (linear density parity check) code for its special properties discussed below. LDPC decodes very efficiently by possibly sacrificing the encoding efficiency; in the encoding process, we apply $n$ linear operations over various subsets of $k$ symbols for generating $n$ coded symbols, the operations can be represented as a encoding matrix; while at decoding side, only a few symbols out of $n$ coded symbols are required to check the validity of those symbols; this can be checked by computing if they produce a zero vector after a linear operations, which is captured in a parity equation. We shown in the section 5 that the scheme we are using with LPDC only use number of symbol of roughly $O(\log b)$ bytes, whereas both 1D-RS and 2D-RS requires the decoder to download all $b$ bytes and $O(\sqrt{b} \log b)$ bytes to run the decoding algorithm assuming symbol size is constant. Erasure Code provides erasure resistance, but there is a limit on how many modified or missing symbols can be tolerated. To capture this property, every erasure code has a parameter called undecodable ratio. The set of symbols that meets the undecodable ratio is called a stopping set, and there might be multiple of such set. In the CMT\cite{yu2020coded} paper, the undecodable ratio is 0.125; both 1D-RS and 2D-RS can produce codes that tolerate any ratio of missing symbols. To compensate high encoding complexity, in Section \ref{sec:implementation} and Appendix \ref{sec:implementation-appendix} we discussed possible ways to parallelize LPDC encoding algorithm for a faster block generation rate. A formal introduction to LDPC can be found at \cite{richardson2008modern} chapter 3.
\section{Bad code handling}
\label{sec:bad-code}

The LDPC codes are probabilistically generated, and some codes may have smaller stopping size (see Appendix \ref{sec:erasure-code}) than designed due to randomness in the generation process. The CMT \cite{yu2020coded} paper section 5.4 table 3 provides the likelihood of such event that a bad code is generated: when the number of symbols, n, in a layer is 256, the probability of bad code is 0.886; when n=512, prob=5.3e-2; when n=1024, prob=2e-3. The calculation process can be found at \cite{yu2020coded} Appendix A. We layout the following protocol for explicitly handling such bad code.

A bad code handling protocol is triggered by an honest side blockchain node who is unable to reconstruct a block after receiving a sufficient number of valid chunks from the oracle layer. The honest side blockchain node sends a bad code message to all oracle nodes containing the block chunks for demonstrating the code is bad. Then all honest oracle nodes communicate with each other, and confirm if the code is bad by exchanging their local chunks to reconstruct the original block. If the code is bad, honest oracle nodes either reconstruct the block but with more chunks or fail to reconstruct the block. The honest oracle nodes then vote and write a proof to the main chain so that a new code will be accepted. The generating process of new codes can be done at each oracle node using a commonly agreed seed, which can be setup in the main chain after bad code proof is voted and accepted. The oracle node then communicates and votes for the new code, and eventually a valid new code is committed to the main chain.


\end{document}